\documentclass[11pt]{article}
\usepackage{fullpage}
\usepackage{tikz}
\usetikzlibrary{arrows.meta,positioning,matrix}
\usepackage{amsmath}
\usepackage{amssymb}
\usepackage{enumitem}
\usepackage{amsthm}
\usepackage[most]{tcolorbox}
\usepackage{subfiles}
\usepackage{mathtools}
\usepackage{algorithm,algpseudocode,float,tcolorbox}
\usepackage{thm-restate,mathrsfs}
\usepackage{abraces}
\usepackage[dvipsnames]{xcolor}
\usepackage{nicematrix}
\usepackage{hyperref}
\hypersetup{
	colorlinks,
	linkcolor={blue!100!black},
	citecolor={red!100!black},
}
\usepackage{cleveref}
\usepackage{dsfont}
\usepackage{blkarray}
\usepackage{booktabs}
\def\01{\{0,1\}}
\usepackage{soul}
\usepackage[most]{tcolorbox}

\newtheorem{theorem}{Theorem}[section]

\newtheorem{lemma}[theorem]{Lemma}

\newtheorem{defi}[theorem]{Definition}

\newtheorem{conjecture}[theorem]{Conjecture}

\newtheorem{proposition}[theorem]{Proposition}

\newtheorem{fact}[theorem]{Fact}

\newcommand{\ket}[1]{\ensuremath{\left|#1\right\rangle}}

\newcommand{\norm}[1]{\ensuremath{\left\|#1\right\|}}

\usetikzlibrary{decorations.pathreplacing}
\newcommand{\Q}{\ensuremath{\mathsf{Q}}}

\newcommand{\R}{\ensuremath{\mathsf{R}}}

\newcommand{\D}{\ensuremath{\mathsf{D}}}

\newcommand{\E}{\mathbb{E}}

\newcommand{\zone}{\{0,1\}}

\newcommand{\cbra}[1]{\left\{#1\right\}}
\newcommand{\rbra}[1]{\left(#1\right)}

\newcommand{\Cc}{\mathcal{C}}

\newcommand{\F}{\mathbb F}

\newcommand{\one}{\mathbf 1}

\newcommand{\C}{\mathsf{C}}
\newcommand{\bs}{\mathsf{bs}}
\newcommand{\fc}{\mathsf{FC}}
\newcommand{\fbs}{\mathsf{fbs}}
\renewcommand{\cal}[1]{\mathcal{#1}}
\newcommand{\fetd}{\textnormal{fETD}}
\newcommand{\etd}{\textnormal{ETD}}

\newcommand{\ETD}{\textnormal{ETD}}
\newcommand{\fgamma}{\textnormal{f}\gamma}
\newcommand{\Ptrunc}{P_{\mathsf{trunc}}}

\usepackage{tikz}
\usetikzlibrary{positioning}

\DeclareMathOperator{\rank}{rank}

\DeclareMathOperator*{\argmax}{arg\,max}

\allowdisplaybreaks

\title{Optimal Quantum-Classical Separations for Exact Learning}

\author{
Srinivasan Arunachalam\\[2mm]
\small IBM Research\\
\small \texttt{Srinivasan.Arunachalam@ibm.com}
\and
Amin Shiraz Gilani\\[2mm]
\small QuICS,\hspace{-0.2mm} University of Maryland\\
\small \texttt{asgilani@umd.edu}
\and
Nikhil S.~Mande\\[2mm]
 \small University of Liverpool  \\
\small \texttt{ mande@liverpool.ac.uk}
}

\date{}

\begin{document}

\maketitle
\begin{abstract}
We study exact learning with membership queries for concept classes \(\mathcal C\subseteq\{0,1\}^N\), focusing on the relationships among their deterministic, randomized, and quantum query complexities, denoted \(\mathsf{D}(\mathcal C)\), \(\mathsf{R}(\mathcal C)\), and \(\mathsf{Q}(\mathcal C)\), respectively. The two canonical quantum speedups in this model are witnessed by Grover search and Bernstein-Vazirani, leading to the longstanding conjecture
$$
\R(\mathcal C)=O(\Q(\mathcal C)^2+\Q(\mathcal C)\log N).
$$
We first refute this conjecture by constructing concept classes $\Cc$ and $\Cc'$ satisfying
\[
\mathsf{R}(\mathcal C)=\Omega\!\left(\frac{\mathsf{Q}(\mathcal C)^3\log N}{\log \mathsf{Q}(\mathcal C)}\right)
\qquad\text{and}\qquad
\mathsf{D}(\mathcal C')=\Omega(\mathsf{Q}(\mathcal C')^3\log N).
\]
The first bound matches the upper bound of Arunachalam et al.~[Quantum'21] up to constant factors, while the second matches the upper bound of Servedio and Gortler~[SICOMP'04]. In particular, this shows that the saving in the randomized upper bound of Arunachalam et al.~fundamentally relies on randomness.
Apart from \emph{characterizing} the optimal relationship between classical and quantum query complexity, our results are the first to show that quantum speedups for learning can go beyond the Grover and Bernstein-Vazirani paradigms.
\end{abstract}

\setcounter{tocdepth}{2}
{\footnotesize \tableofcontents}

\newpage

\section{Introduction}
Exact learning with membership queries is a foundational model of active learning, originating in the seminal work of Angluin~\cite{Angluin88}. A \emph{concept class} is a collection of Boolean functions
$    \Cc \subseteq \{c:\01^n\to\01\}.$
Equivalently, identifying each Boolean function with its truth table and writing \(N=2^n\), we may view \(\Cc\) as a subset of \(\01^N\). The learner is given query access to an unknown target concept \(c^*\in\Cc\), with the goal of identifying \(c^*\).
 In the classical model, a membership query at \(x\in\01^n\) returns the value \(c^*(x)\). A deterministic learner must identify \(c^*\) with certainty, whereas a randomized learner is required to succeed with probability at least \(2/3\). A quantum learner is instead given quantum query access to \(c^*\), allowing superposition queries via the unitary
$$
    \ket{x,b}\longmapsto \ket{x,b\oplus c^*(x)}.
$$
We denote by \(\D(\Cc)\), \(\R(\Cc)\), and \(\Q(\Cc)\) the deterministic, randomized, and bounded-error quantum query complexities of identifying an unknown concept in \(\Cc\), respectively.

The classical membership-query model has been studied extensively both within and beyond computational learning theory~\cite{BshoutyCleveGavaldaKannanTamon96,Bshouty13,Bshouty18,BshoutyH19,hopkins2020point,chase2020bounds,Bshouty25}. More broadly, it captures the fundamental black-box identification problem of determining an unknown function promised to belong to a known family, with connections to interpolation, functional verification, and black-box identity testing~\cite{Bshouty13}. Roughly speaking, the query complexity of exact learning measures the amount of information that must be extracted from an unknown function in order to identify it completely. The quantum analogue has likewise received substantial attention~\cite{servedio2004equivalences,atici2005improved,ambainis2004quantum,Kothari14,arunachalam2017guest,arunachalam2021two}. In this setting, also studied under the name \emph{oracle identification}, a basic question is how much quantum queries can reduce the identification of an unknown concept. This leads to a natural question first studied by Servedio and Gortler~\cite{servedio2004equivalences}:

\begin{tcolorbox}
\begin{quote}
\centering
\emph{What is the optimal relationship between classical and quantum  learning?}
\end{quote}
\end{tcolorbox}
Servedio and Gortler~\cite{servedio2004equivalences} made the first progress by proving the following theorem.

\begin{theorem}[{\cite[Theorem~1]{servedio2004equivalences}}]\label{thm:sg}
For all $\mathcal{C}\subseteq\{0,1\}^N$, we have
$\mathsf{D}(\mathcal{C}) = O\!\left(\mathsf{Q}(\mathcal{C})^3\log N\right).$ 
\end{theorem}

Later, Arunachalam et al.~\cite{arunachalam2021two} improved their simulation, albeit in the randomized setting.

\begin{theorem}[{\cite[Section~4]{arunachalam2021two}}]\label{thm:twonew}
    For all $\cal{C} \subseteq \zone^N$, we have
$
\mathsf{R}(\mathcal{C}) = O\!\left(\frac{\mathsf{Q}(\mathcal{C})^3}{\log \mathsf{Q}(\mathcal{C})} \log N \right).
$
\end{theorem}

Apart from these simulations, the best known separations between quantum and classical learning was by two well-known classes: $(i)$ point functions, $\Cc=\{e_i:i\in [N]\}$, for which the randomized query complexity is $\Omega(N)$ and the quantum query complexity is $O(\sqrt{N})$ by Grover's algorithm~\cite{grover1996fast} and $(ii)$ parity functions, $\Cc=\{\chi_S:\chi_S(x)=(-1)^{S\cdot x}\text{ for all } x\}_{S \subseteq [N]}$, for which the randomized query complexity is $O(\log N)$ and the quantum query complexity is $O(1)$ by Fourier sampling~\cite{bernstein1997quantum}. Apart from these two generic separations and reasonably natural generalizations of them~\cite{atici2005improved}, no  non-trivial separation was known between quantum and classical learning which also accounted for the domain size.\footnote{We remark that Belovs~\cite{belovs2015quantum} (and also Atici and Servedio~\cite{atici2005improved}) gave a quartic separation between quantum and classical learning, however, these classes are consistent with Conjecture~\ref{conj:aticiservedio}: the concepts are functions on \(n\) variables, and the \(Q(\Cc)\log N\) term in the conjectured upper bound is sufficiently large to accommodate the separation. }  This led to the natural question, raised and conjectured in several works~\cite{servedio2004equivalences,atici2005improved,montanaro2007structure,arunachalam2017guest,arunachalam2021two}, with its earliest explicit formulation appearing two decades ago in~\cite[Question~4.4]{atici2005improved}.

    \begin{conjecture}
    \label{conj:aticiservedio}
        For every concept class $\Cc\subseteq \{c:\01^n\rightarrow \01\}$, we have
    \[
 \R(\cal{C})=O\!\left(\Q(\cal{C})^2+\Q(\cal{C})\log N\right).
    \]\end{conjecture}
   This bound was believed to be tight for the point function and the parity function class.
\paragraph{Main result}
Our main contribution in this work is in refuting this conjecture.\footnote{The class $\Cc=\{\chi_s:\chi_s(x)=(-1)^{S\cdot x}\text{ for all } x\}_s$ of parity functions witnesses the bounds in our results, but has $\Q(\Cc) = 1$, so this is not a very interesting example. The question is thus interesting in the regime where the underlying concept class has $\Q(\cal C) = \omega(1)$.}

\begin{restatable}{theorem}{dqseparation}\label{thm:dqseparation}
There exists $\cal C\subseteq\zone^N$ with $\Q(\cal C)=\omega(1)$ and $\D(\cal C)=\Omega(\Q(\cal C)^3\log N)$.
\end{restatable}

\begin{restatable}{theorem}{rqseparation}\label{thm:rqseparation}
There exists $\cal C\subseteq\zone^N$ with $\Q(\cal C)=\omega(1)$ and $\R(\cal C)=\Omega\!\left(\frac{\Q(\cal C)^3}{\log\Q(\cal C)}\log N\right)$.
\end{restatable}
Putting these together with Theorem~\ref{thm:sg} and Theorem~\ref{thm:twonew}, our work gives the optimal relation between quantum and classical learning, both in the deterministic and randomized setting, upto constant factors. In particular, it refutes the two-decade-old conjecture (\Cref{conj:aticiservedio}) of At{\i}ci and~Servedio~\cite{atici2005improved}.

\subsection{Deterministic separation}

\subsubsection{High-level overview}\label{sec:det-high-level}
We begin with the deterministic setting, where the idea behind the separation is particularly simple. The two canonical quantum advantages in exact learning come from Grover search~\cite{grover1996fast, bennett1997strengths} and Bernstein-Vazirani~\cite{bernstein1997quantum}: Grover gives a quadratic speedup for locating a hidden marked item, while Bernstein-Vazirani learns a hidden parity with one quantum query compared to the logarithmic classical lower bound.

Our construction combines these two advantages so that the classical costs multiply while the quantum costs add. We hide a Bernstein-Vazirani-amenable object at one of many possible locations. A deterministic learner must spend many queries ruling out such an object at each location, whereas a quantum learner can first locate the correct location using Grover search and then recover the hidden object using Bernstein-Vazirani.
 Our main observation is that a bilinear form provides exactly these two properties: classical queries reveal only one constraint at a time, while quantum queries can access many of its coordinates simultaneously. Fix $q\geq 2$. Let $c_0$ denote the all-zero function on $[q]^2\times(\mathbb F_2^q)^2$, and for each $b\in[q^2]$ and nonzero $A\in\mathbb F_2^{q\times q}$ define
\begin{equation}\label{eq:cba}
c_{b,A}(b',x,y)
:= \underbrace{\mathbf 1[b'=b]}_{\text{Grover structure}}
\cdot
\underbrace{x^\top Ay}_{\text{Bilinear structure}}.
\end{equation}
Our concept class is
\begin{equation}\label{eq:det-concept-class}
\mathcal C:=\{c_0\}\cup\{c_{b,A}:b\in[q^2],\ A\in\mathbb F_2^{q\times q}\setminus\{0\}\}.
\end{equation}
For a concept $c_{b,A}$, all $q^2-1$ ``blocks'' $b'\neq b$ are identically zero, while the unique hidden block $b$ contains the bilinear form associated with $A$. Quantumly, the hidden block can be located using $O(q)$ queries, after which $A$ can be recovered using another $O(q)$ queries. Deterministically, however, each zero answer within a block imposes only one linear constraint on the $q^2$ entries of $A$, so ruling out a block can require $q^2$ queries. Since there are $q^2$ possible blocks, this leads to an $\Omega(q^4)$ deterministic lower bound, compared with an $O(q)$ quantum upper bound. Finally, $N=q^2 2^{2q}$, so $\log N=\Theta(q)$, and therefore $\D(\mathcal C)=\Omega(\Q(\mathcal C)^3\log N)$, proving \Cref{thm:dqseparation}.

\subsubsection{Quantum and classical bounds}
\textbf{Quantum upper bound.} Our quantum learner for the concept class described in \Cref{eq:cba} and \Cref{eq:det-concept-class} first locates the hidden block, then learns $A$.

 \emph{Learning the hidden block.}
    Let $A_{b'}=A \cdot \mathbf 1 [b=b']$.
The algorithm prepares a uniform superposition over \(b',x,y\) and use one membership query as a phase query. Up to normalization, this produces
$$
    \sum_{b'\in[q^2]}
    \sum_{x,y\in\mathbb{F}_2^q}
       (-1)^{x^\top A_{b'}y}
       \ket{b',y,x}.
$$
Applying a Hadamard transform to the \(x\)-register and using the orthogonality of characters the new state is
$$
    \sum_{b'\in[q^2]}
    \sum_{y\in\mathbb{F}_2^q}
       \ket{b',y,A_{b'}y}.
$$
Thus every basis state corresponding to a zero block has $0$ in its final register, whereas for the hidden block the final register contains $Ay$.
Since $A\neq 0$, we have $\rank(A)\geq 1$, and hence
\[
\Pr_{y\in\mathbb F_2^q}[Ay\neq 0]
=1-2^{-\rank(A)}
\geq \frac12.
\]
Since only one of the \(q^2\) blocks is nonzero, the total probability mass on states with final register nonzero is at least
$1/2q^2$. 
Amplitude amplification therefore boosts the probability of observing a basis state with nonzero final register to a constant using $O(q)$ membership queries.  The reflections required for amplitude amplification are efficiently implementable, since the good subspace is determined by whether the final register is nonzero and the initial state is prepared by a known unitary using one membership query. 
We then measure. If the final register is nonzero, the first register reveals the hidden block $b$; otherwise, the learner outputs $c_0$.

\emph{Learning $A$.}
Once $b$ is known, we recover $A$ column by column using Bernstein-Vazirani. This is done as follows: For each $j\in[q]$, fix $y=e_j$ and consider the function
\[
x\mapsto c_{b,A}(b,x,e_j)=x^\top A e_j.
\]
Since $x^\top A e_j=\langle x,Ae_j\rangle$, this is precisely a Bernstein-Vazirani instance with hidden string $Ae_j$. Hence one Bernstein-Vazirani query recovers the $j$th column of $A$. Repeating this for $j=1,\ldots,q$ recovers the entire matrix using $q$ additional queries. Overall,
\[
\Q(\Cc)=O(q).
\]

\paragraph{Deterministic lower bound.}
We prove the lower bound using an adversary argument. Consider a deterministic learner making fewer than $q^4$ queries, and consider the transcript obtained by answering $0$ to every query. Since there are $q^2$ blocks, by averaging there exists a block $b^\star$ that receives fewer than $q^2$ queries. Let $(b^\star,x_1,y_1),\ldots,(b^\star,x_t,y_t)$
be the queries made in block $b^\star$, where $t<q^2$. We claim that there exists a nonzero matrix $A^\star\in\mathbb F_2^{q\times q}$ such that $x_s^\top A^\star y_s=0$ for every $s\in[t]$. Each condition $x_s^\top A^\star y_s=0$ is a homogeneous linear equation in the $q^2$ entries of $A^\star$. Since there are only $t<q^2$ such equations, the resulting homogeneous system has a nonzero solution.
Therefore the all-zero transcript is consistent with both $c_0$ and $c_{b^\star,A^\star}$: queries outside block $b^\star$ are answered $0$ automatically, while the choice of $A^\star$ ensures that every query inside block $b^\star$ is also answered $0$. Hence no deterministic learner making fewer than $q^4$ queries can identify the target concept, and so $\D(\Cc)\geq q^4$. Combining this with $\Q(\Cc)=O(q)$ and $\log N=\Theta(q)$, where the latter follows from $N=q^2 2^{2q}$, we obtain
\[
\D(\Cc)
=\Omega(q^4)
=\Omega\!\left(\Q(\Cc)^3\log N\right).
\]
This proves Theorem~\ref{thm:dqseparation}.
\subsection{Randomized separation}
\subsubsection{High-level overview}
First observe that the deterministic separation above does not directly yield a randomized separation. Indeed, consider the concept class $\Cc$ described in \Cref{eq:cba} and \Cref{eq:det-concept-class}. For uniformly random $x,y\in\mathbb F_2^q$, we have $\Pr[x^\top Ay=1]=\frac12(1-2^{-\rank(A)})\geq \frac14$ whenever $A\neq 0$.
Thus, for each $b'\in[q^2]$, a randomized learner can test whether $b'=b$ using $O(1)$ random queries to block $b'$. All zero blocks always return $0$, while the hidden block is detected with constant probability. Once the hidden block $b$ is found, the learner can recover $A$ using $q^2$ queries: for every $i,j\in[q]$, query $(e_i,e_j)$ in block $b$, which returns $e_i^\top A e_j=A_{ij}$.
Thus all entries of $A$ can be recovered using $q^2$ queries. Hence the previous class satisfies $\R(\Cc)=O(q^2)$ in contrast with the $\D(\Cc) = \Omega(q^4)$ lower bound outlined in the previous section.

As a first step towards our randomized separation, we modify the construction above, which yields a concept class $\Cc$ that already refutes Conjecture~\ref{conj:aticiservedio} by obtaining
\[
\R(\Cc)=\Omega(\Q(\Cc)^2\log N).
\]
We discuss this candidate as a warm-up in Section~\ref{sec:warmup}. However, pushing beyond this requires new ideas. We discuss a second issue that we had to circumvent for our separation.
Recall that~\cite{arunachalam2021two} showed the upper bound of
\[
    \R(\Cc)=O\!\left(\frac{\Q(\Cc)^3\log N}{\log \Q(\Cc)}\right),
\]
for all concept classes $\Cc$.  To obtain our stronger separation, we aim for a concept class satisfying
\[
\R(\Cc)=\Omega(\Q(\Cc)^3)
\qquad\text{and}\qquad
\log N=\Theta(\log \Q(\Cc)).
\]
The second requirement prevents us from using Bernstein-Vazirani in the same way as in the deterministic separation. Indeed, obtaining a factor-$t$ classical cost from a parity on $t$ bits requires a domain of size $2^t$, and hence forces $\log N=\Omega(t)$. We therefore need a different inner problem that retains a one-versus-$t$ quantum-classical gap while having domain size only polynomial in $t$.

Our main idea is to move away from Bernstein-Vazirani as the main source of the separation and instead use a problem inspired by the hidden subgroup problem. In particular, consider the following \emph{hidden line problem over a finite field}.

Let $t\ge4$ be a power of two and let $s\in \mathbb F_{t^6}$ be unknown. Define $F_s:\mathbb F_{t^6}^{\,2}\to \mathbb F_{t^6}$ by
\[
F_s(x,c)=P(c+xs),
\]
where $P$ is a polynomial whose precise form will be specified later. The function is constant on each affine line $c+xs=y$, with the hidden parameter $s$ determining the slope of these lines. If one had direct query access to $F_s$, then one quantum query followed by Fourier sampling would recover $s$ with high probability. Our Boolean concept will encode these hidden-line values using two parts. In the first part, each pair $(x,c)\in \mathbb F_{t^6}^{\,2}$ indexes a block containing a unique marked address encoding $F_s(x,c)$. The second part contains auxiliary information that allows the learner to recover $P$ efficiently once $s$ is known.

Fix a block $(x,c)$. It contains $t^2$ possible addresses and exactly one of them is marked, namely $P_{\mathsf{trunc}}(c+xs)$. Thus, classically, uncovering the marked address in this block is an unstructured search problem of finding a unique marked element among $t^2$ possibilities, which is well known to require $\Omega(t^2)$ queries. Quantumly, the same marked address can be found in $O(t)$ queries using Grover search. Moreover, for any $t$ blocks whose associated hidden inputs are distinct, their marked addresses are independent and uniform. Thus, heuristically, uncovering the marked addresses in $t$ such blocks requires about $t\cdot t^2=t^3$ queries, suggesting an $\Omega(t^3)$ randomized lower bound. There are, however, a few nontrivial issues in turning this intuition into the desired separation.
\begin{enumerate}
    \item \textbf{Hiding the marked addresses.} For the randomized lower bound, it is not enough merely to hide each useful value at one of $t^2$ locations: the locations of the marked addresses must remain difficult to predict even for an adaptive learner. We therefore choose $P$ at random from a suitable family of degree-$t$ polynomials. The randomness in the coefficients ensures that the values of $P$ at any $t$ distinct inputs are independent and uniform. Consequently, after seeing fewer than $t$ such values, the marked address corresponding to a new input remains uniform and unpredictable.

    \item \textbf{Hiding the slope.} Even after a randomized learner discovers some marked addresses, the corresponding values should not determine the hidden slope $s$ too quickly. As long as fewer than $t$ distinct underlying inputs have been exposed, the corresponding polynomial values remain independent and uniform. We then show that collisions among these underlying inputs are unlikely, so a small number of revealed values gives essentially no information about $s$.
    
    \item \textbf{Recovering the polynomial.} The previous two issues concern the randomized lower bound. On the quantum side, recovering the hidden slope $s$ is not enough to identify the concept, since the polynomial $P$ remains unknown. We therefore add an auxiliary ``cheat sheet'' component whose unique nonzero sheet is indexed by $s$. Once $s$ has been recovered, the quantum learner can access this sheet and efficiently recover the coefficients of $P$. All other sheets are identically zero, and since there are many possible sheet addresses, a randomized learner cannot locate the active sheet without essentially learning $s$ first.
\end{enumerate}

With this motivation, we now formally define the concept class.
Fix a binary basis $e_1,\ldots,e_{6\log t}$ of $\mathbb F_{t^6}$, and let $\mathsf{trunc}:\mathbb F_{t^6}\to \{0,1\}^{2\log t}$ denote truncation to the first $2\log t$ coordinates. 
Since each element of $\{0,1\}^{2\log t}$ has exactly $t^4$ preimages under $\mathsf{trunc}$, truncation maps a uniform element of $\mathbb F_{t^6}$ to a uniform element of $\{0,1\}^{2\log t}$.
An unknown concept is indexed by a pair $(P,s)$, where $s\in \mathbb F_{t^6}$ and
$P\in \mathbb F_{t^6}[X]$ is the monic degree-$t$ polynomial
\begin{equation}\label{eq:p-props}
P(X)=X^t+\sum_{j=0}^{t-1}a_jX^j,
\qquad
a_0,\ldots,a_{t-1}\in \mathbb F_{t^6}.
\end{equation}
Thus $P$ defines a function $P:\mathbb F_{t^6}\to \mathbb F_{t^6}$. Write $P_{\mathsf{trunc}}:=\mathsf{trunc}\circ P:\mathbb F_{t^6}\to \{0,1\}^{2\log t}$. For each pair $(P,s)$ with $s \in \mathbb F_{t^6}$ and $P$ as in \Cref{eq:p-props}, the concept $c_{P,s}$ consists of two parts.
\begin{itemize}
    \item \textbf{Block part.} For every pair $(x,c)\in \mathbb F_{t^6}^{\,2}$, the coordinates $\{(x,c,z):z\in \{0,1\}^{2\log t}\}$ form a \emph{block} of size $t^2$. Exactly one coordinate in this block is \emph{marked}, namely the coordinate with $z=P_{\mathsf{trunc}}(c+xs)$. The concept takes value $1$ at this coordinate and $0$ at every other coordinate in the block.

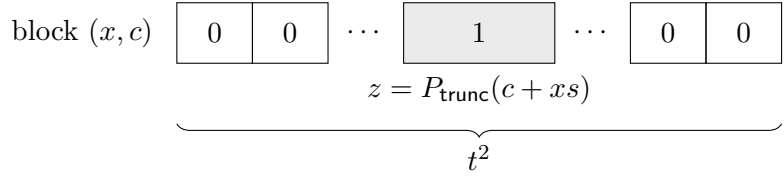
\begin{figure}[!ht]
\centering
\begin{tikzpicture}[x=1cm,y=1cm]
\node at (-1.25,0.4) {$\text{block }(x,c)$};
\draw (0,0) rectangle (1,0.8);
\draw (1,0) rectangle (2,0.8);
\draw[fill=black!8] (3,0) rectangle (5,0.8);
\draw (6,0) rectangle (7,0.8);
\draw (7,0) rectangle (8,0.8);
\node at (0.5,0.4) {$0$};
\node at (1.5,0.4) {$0$};
\node at (2.5,0.4) {$\cdots$};
\node at (4,0.4) {$1$};
\node at (5.5,0.4) {$\cdots$};
\node at (6.5,0.4) {$0$};
\node at (7.5,0.4) {$0$};
\node at (4,-0.35) {$z=P_{\mathsf{trunc}}(c+xs)$};
\draw[decorate,decoration={brace,mirror,amplitude=5pt}]
    (0,-0.8) -- (8,-0.8)
    node[midway,yshift=-0.45cm] {$t^2$};
\end{tikzpicture}
\caption{Each block has one marked address, namely $z=P_{\mathsf{trunc}}(c+xs)$. Finding the marked address reveals the value $P_{\mathsf{trunc}}(c+xs)$.}\label{fig:block}
\end{figure}
There are $t^{12}$ choices of $(x,c)$, and each block contains $t^2$ coordinates. Hence the block part~contains
$t^{12}\cdot t^2=t^{14}$ 
bits in total. See \Cref{fig:block} for a visual description of the block part.
    \item \textbf{Cheat-sheet part.} Recovering $s$ alone does not identify the concept, since the polynomial $P$ remains unknown. We therefore include one auxiliary sheet for each $u\in \mathbb F_{t^6}$. Every sheet with $u\neq s$ is identically zero, while the unique active sheet $u=s$ contains Hadamard encodings of the coefficients $a_0,\ldots,a_{t-1}$ of $P$ described in \Cref{eq:p-props}. Concretely, the coordinate indexed by $(u,j,b)$, where $j\in\{0,\ldots,t-1\}$ and $b\in\mathbb F_{t^6}$, has value
\[
\mathbf 1[u=s]\cdot \langle b,a_j\rangle.
\]
Once $s$ is known, the learner knows which sheet is active and can recover the coefficients of $P$ efficiently. There are $t^6$ choices of $u$, $t$ choices of $j$, and $t^6$ choices of $b$. Hence the cheat-sheet part contains
$t^6\cdot t\cdot t^6=t^{13}$ 
bits in total. See \Cref{fig:cheat} for a visual description of the cheat-sheet~part.
    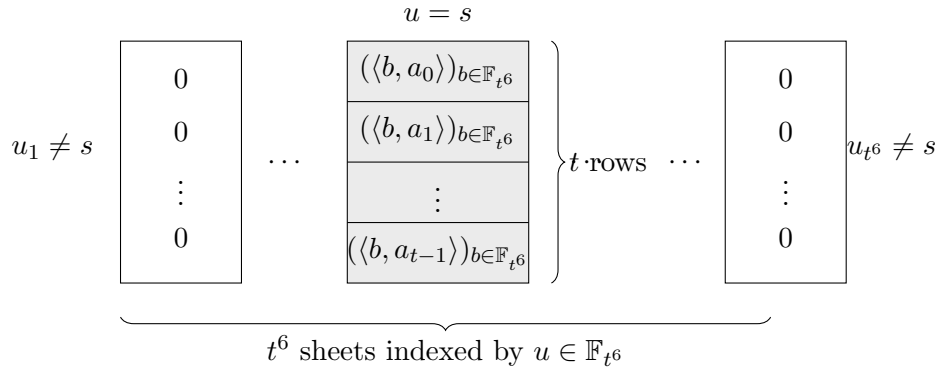
\begin{figure}[!ht]
\centering
\begin{tikzpicture}[x=1cm,y=1cm]

\node at (-0.9,1.8) {$u_1\neq s$};
\draw (0,0) rectangle (1.6,3.2);
\node at (0.8,2.7) {$0$};
\node at (0.8,2.0) {$0$};
\node at (0.8,1.3) {$\vdots$};
\node at (0.8,0.6) {$0$};

\node at (2.2,1.6) {$\cdots$};

\node at (4.2,3.55) {$u=s$};
\draw[fill=black!8] (3,0) rectangle (5.4,3.2);
\draw (3,2.4) -- (5.4,2.4);
\draw (3,1.6) -- (5.4,1.6);
\draw (3,0.8) -- (5.4,0.8);

\node at (4.2,2.8) {$(\langle b,a_0\rangle)_{b \in \mathbb F_{t^6}}$};
\node at (4.2,2.0) {$(\langle b,a_1\rangle)_{b \in \mathbb F_{t^6}}$};
\node at (4.2,1.2) {$\vdots$};
\node at (4.2,0.4) {$(\langle b,a_{t-1}\rangle)_{b \in \mathbb F_{t^6}}$};

\node at (6.2,1.6) {$\cdots$};

\node at (10.2,1.8) {$u_{t^6}\neq s$};
\draw (8,0) rectangle (9.6,3.2);
\node at (8.8,2.7) {$0$};
\node at (8.8,2.0) {$0$};
\node at (8.8,1.3) {$\vdots$};
\node at (8.8,0.6) {$0$};

\node at (7.5,1.6) {$\cdots$};
\draw[decorate,decoration={brace,mirror,amplitude=5pt}]
    (0,-0.45) -- (8.6,-0.45)
    node[midway,yshift=-0.45cm] {$t^6$ sheets indexed by $u\in \mathbb F_{t^6}$};

\draw[decorate,decoration={brace,mirror,amplitude=5pt}]
    (5.7,0) -- (5.7,3.2)
    node[midway,xshift=0.75cm] {$t$ rows};

\end{tikzpicture}
\caption{The cheat-sheet part consists of \({t^6}\) sheets indexed by \(u\in \mathbb F_{t^6}\). All sheets are zero except the active sheet \(u=s\), whose \(j\)-th row is the  Hadamard encoding \((b\mapsto \langle b,a_j\rangle)_{b \in \mathbb F_{t^6}}\) of the coefficients~\(a_j\).}\label{fig:cheat}
\end{figure}
\end{itemize}

The concept class  $\mathcal C_t$ is now implicitly defined to be the set of all Boolean-valued concepts, one for each $P,s$ defined as above.

\subsubsection{Quantum and classical bounds}
\paragraph{Quantum upper bound.}
The quantum learner first recovers the hidden slope $s$ from the block part. For every $(x,c)\in \mathbb F_{t^6}^{\,2}$, exactly one of the $t^2$ addresses is marked, namely $z=P_{\mathsf{trunc}}(c+xs)$. Using exact amplitude amplification over the address register, the learner can therefore prepare
\[
\frac1{t^6}\sum_{x,c\in \mathbb F_{t^6}}\ket{x,c,P_{\mathsf{trunc}}(c+xs)}
\]
using $O(t)$ membership queries. Importantly, this can be done coherently with $(x,c)$ in superposition.

Next, apply $H^{\otimes 6\log t}$ to each of the first two registers and measure them. Writing $y=c+xs$ and using orthogonality of characters, the resulting outcome $(\alpha,\beta)$ satisfies $\langle\alpha,x\rangle=\langle\beta,xs\rangle$ for every $x\in \mathbb F_{t^6}$.
Whenever $\beta\neq 0$, taking $x=e_i$ for $i=1,\ldots,6\log t$ gives $\alpha_i=\langle\beta,e_i s\rangle$, which is a system of $6\log t$ linear equations over $\F_2$ that uniquely determines $s$.
Indeed, suppose that $s'$ is another solution. Then
\[
\langle\beta,x(s-s')\rangle=0
\qquad\text{for every }x\in \mathbb F_{t^6}.
\]
If $s\neq s'$, multiplication by $s-s'$ is a bijection on $\mathbb F_{t^6}$, and hence
$\langle\beta,y\rangle=0$ for every $y\in \mathbb F_{t^6}$, which implies $\beta=0$, a contradiction. Thus, whenever $\beta\neq0$, the outcome $(\alpha,\beta)$ uniquely determines $s$. 

It remains to bound the probability of the exceptional case $\beta=0$. For every label $z\in \{0,1\}^{2\log t}$, there are exactly $t^4$ elements $w\in \mathbb F_{t^6}$ satisfying $\mathsf{trunc}(w)=z$. Since $P$ has degree $t$, each equation $P(y)=w$ has at most $t$ solutions. Hence
\[
|P_{\mathsf{trunc}}^{-1}(z)|
\le t\cdot t^4
=t^5.
\]

For $\beta=0$, necessarily $\alpha=0$. For each $z\in \{0,1\}^{2\log t}$, the amplitude of $\ket{0,0,z}$ is
\[
\frac{1}{t^{12}}\bigl|\{(x,c)\in \mathbb F_{t^6}^{\,2}:P_{\mathsf{trunc}}(c+xs)=z\}\bigr|
=
\frac{|P_{\mathsf{trunc}}^{-1}(z)|}{t^6},
\]
since for each $x$ there are exactly $|P_{\mathsf{trunc}}^{-1}(z)|$ choices of $c$. Therefore
\[
\Pr[\beta=0]
=
\frac1{t^{12}}\sum_{z\in \{0,1\}^{2\log t}}|P_{\mathsf{trunc}}^{-1}(z)|^2
\le
\frac1{t^{12}}
\left(\max_{z\in \{0,1\}^{2\log t}}|P_{\mathsf{trunc}}^{-1}(z)|\right)
\sum_{z\in \{0,1\}^{2\log t}}|P_{\mathsf{trunc}}^{-1}(z)|
\le \frac1t,
\]
where we used $\sum_{z\in \{0,1\}^{2\log t}}|P_{\mathsf{trunc}}^{-1}(z)|={t^6}$. Thus the learner recovers $s$ with probability at least $1-1/t$ using $O(t)$ membership queries.
Finally, once $s$ is known, the learner knows which cheat sheet to
query. Its $t$ rows contain Hadamard encodings of the
coefficients of $P$. Each coefficient can be recovered with one query,
so all of $P$ is learned using another $t$ queries. Therefore, the
quantum query complexity is
\[
\underbrace{O(t)}_{\substack{\text{prepare the}\\\text{label superposition}}}
\cdot
\underbrace{1}_{\substack{\text{one Fourier}\\\text{sample}}}
\;+\;
\underbrace{O(t)}_{\substack{\text{recover }P\\\text{from cheat sheet }s}}
=O(t).
\]

\paragraph{Randomized lower bound.}
We prove that \(\R(\cal C_t)=\Omega(t^3)\) using Yao's minimax principle and a sequence of hybrid experiments. Choose \(s,a_0,\ldots,a_{t-1}\) independently and uniformly from \(\F_{t^6}\), and fix a deterministic learner making at most \(D\) queries. We compare its execution under this distribution with three successively simpler experiments. Each hybrid removes one source of dependence on the hidden slope \(s\), while changing the learner's success probability by only a small amount. In the final experiment, the learner's entire transcript is independent of \(s\), at which point the desired lower bound is immediate. We now present the three hybrids and summarize the eventual lower bound afterwards.

\begin{enumerate}
    \item \textbf{We may assume the active cheat sheet is never queried.}
    First modify the execution by answering every cheat-sheet query with $0$, while leaving all block answers unchanged. Let $A$ be the event that, in this modified execution, the learner either queries the sheet indexed by $s$ or outputs a concept with slope $s$.
    
    The real and modified executions agree until the learner first queries the sheet indexed by $s$. Hence, if the learner succeeds in the real execution, then either it must query this sheet, or else the two executions remain identical throughout and its final output must have the correct slope. Therefore,
    \[
    \Pr[\mathrm{success}]
    \le
    \Pr[A].
    \]
    \item \textbf{Replacing the polynomial labels by a random function.}
    For any $r\le t$ distinct points $y_1,\ldots,y_r\in \mathbb F_{t^6}$, the values $P(y_1),\ldots,P(y_r)$ are independent and uniform in $\mathbb F_{t^6}$. Hence $P_{\mathsf{trunc}}(y_1),\ldots,P_{\mathsf{trunc}}(y_r)$ are independent and uniform in $\{0,1\}^{2\log t}$.
    
    For fixed $s$, every block query is an equality query of the form $    \mathbf 1[P_{\mathsf{trunc}}(c+xs)=z]$. We show that replacing $P_{\mathsf{trunc}}$ by a fully random function $g:\mathbb F_{t^6}\to \{0,1\}^{2\log t}$ changes the probability of the event $A$ from the previous step by at most
    \[
    \varepsilon_D
    =
    \binom Dt\left(\frac{2}{t^2}\right)^t.
    \]
    For $D$ a sufficiently small constant times $t^3$, this error is exponentially small in $t$.
    \item \textbf{Replacing the random function by independent block labels.}
    After the previous step, the label of block $(x,c)$ is $g(c+xs)$, where $g:\mathbb F_{t^6}\to \{0,1\}^{2\log t}$ is a fully random function. We compare this with the \emph{independent-block experiment} in which each block $(x,c)$ is assigned its own independent uniform label $U_{x,c}\in \{0,1\}^{2\log t}$, reused on repeated queries to that block.
    
    The two experiments can be coupled to agree until two distinct queried blocks $(x,c)$ and $(x',c')$ satisfy 
    \[
    c+xs=c'+x's.
    \]
    Indeed, as long as all these hidden inputs are distinct, the corresponding values of the random function $g$ are themselves independent and uniform in $\{0,1\}^{2\log t}$.
    For any fixed pair of distinct blocks, the above equality holds for at most one slope $s$: if $x=x'$, it is impossible since then $c\neq c'$, while if $x\neq x'$, it requires $s=(c'-c)/(x-x')$. In the independent-block experiment, all answers are independent of $s$, and hence so are the learner's adaptive choices of which blocks to query. Conditioning on the resulting transcript therefore fixes at most $D$ queried blocks while leaving $s$ uniform in $\mathbb F_{t^6}$. A union bound over pairs shows that a collision occurs with probability at most $\binom D2/{t^6}$. Thus this replacement changes the probability of $A$ by at most $\binom D2/{t^6}$.
\end{enumerate}
Thus, after the three hybrid steps, we have reduced to an experiment in which the learner's entire view is independent of the hidden slope \(s\). The only remaining way to identify \(s\) is therefore to guess it, either through one of its cheat-sheet queries or in its final output. 
 Hence the cheat-sheet addresses it queries, as well as its final slope guess, are independent of $s$. There are at most $D+1$ such guesses, so the probability of the event $A$ is at most $(D+1)/{t^6}$. Combining the three steps~gives
\[
\Pr[\mathrm{success}]
\le \frac{D+1}{t^6} + \binom Dt\left(\frac2{t^2}\right)^t + \frac{\binom D2}{t^6}.
\]
For $D=\lfloor t^3/10\rfloor$, this is at most $O(t^{-3})+(e/5)^t+1/200<2/3$ for sufficiently large $t$. Therefore, by Yao's principle,
\[
\R(\mathcal C_t)
\geq 
\underbrace{t}_{\text{finding marked address}}
\cdot
\underbrace{t^2}_{\text{queries per address}}
=
\Omega(t^3).
\]
The $t^{12}$ blocks contribute $t^{14}$ coordinates and the $t^6$ cheat sheets contribute $t^{13}$ coordinates, giving domain size $N=t^{14}+t^{13}=\Theta(t^{14})$, while the $t^{6t}$ choices of $P$ and $t^6$ choices of $s$ define distinct concepts, so $|\mathcal C_t|=t^{6(t+1)}$. Hence the standard exact-learning lower bound of~\cite{servedio2004equivalences} gives
\[
\Q(\mathcal C_t) = \Omega\!\left(\frac{\log |\mathcal C_t|}{\log N}\right) = \Omega(t).
\]
Together with the quantum upper bound, this yields $\Q(\mathcal C_t)=\Theta(t)$ and $\log N=\Theta(\log \Q(\mathcal C_t))$. Thus,
\[
\R(\mathcal C_t) = \Omega\!\left(\frac{\Q(\mathcal C_t)^3\log N}{\log \Q(\mathcal C_t)}\right),
\]
proving \Cref{thm:rqseparation} and matching the general simulation upper bound of~\cite{arunachalam2021two} for this family.

\subsection{Other results}
We now present other structural results on exact learning with membership queries, that might be of independent interest.

\subsubsection{Booleanization of quantum learnability}
Every Boolean decision about an unknown concept is at most as hard as identifying the concept itself: one can first learn the concept and then evaluate the desired function with no extra queries. A natural question is to ask when the opposite direction is true: for what concept classes $\Cc \subseteq \zone^N$ is it true that some Boolean decision witnesses a matching quantum query lower bound? Equivalently, 
\begin{quote}
    \emph{For what promise domains $D\subseteq\zone^N$ does there exist a partial Boolean function $f:D\to\zone$ whose quantum query complexity is asymptotically equal to that of identifying an unknown string in $D$?}
\end{quote}
For the full domain $D=\zone^N$, any maximally hard function for quantum query complexity (e.g., parity) easily witnesses such a lower bound, but it is not clear whether this remains true for an arbitrary restricted domain.
 Prior work used Boolean decisions to lower-bound exact learning~\cite{ambainis2004quantum} and studied learning with respect to prescribed partitions~\cite{atici2005improved}. We show that the answer to the question above is affirmative for \emph{every} promise domain: every concept class admits a Boolean decision whose quantum query complexity matches that of learning the concept itself, i.e.,
\[
\Q(\mathcal C)=\Theta\!\left(\max_{P_1\uplus P_2=\mathcal C}\Q_{\{P_1,P_2\}}(\mathcal C)\right).
\]
We defer the proof to Section~\ref{sec:booleanization}. In contrast, we show that the analogous statement for randomized query complexity is false in a strong sense.

\subsubsection{A new combinatorial measure}
The extended teaching dimension (\(\ETD\)) was introduced by Heged\H{u}s~\cite{hegedHus1995generalized} in connection with the query complexity of exact learning, building on earlier work of Moshkov~\cite{moshkov1982conditional}. In particular, \(\ETD\) and its variants have been used to obtain upper bounds on the number of membership queries required for exact learning~\cite{hegedHus1995generalized,balcazar2002new,BshoutyMakhoul18,hanneke2024star}. In a rather different line of work, Servedio and Gortler~\cite{servedio2004equivalences} introduced the \(\gamma\)-parameter, which had already appeared implicitly as a lower-bound measure in~\cite{bshouty1994oracles}. They showed that
\begin{equation}\label{eq:1/sqrt gamma lower bound}
    \Q(\cal C)
    =
    \Omega\!\left(\frac{1}{\sqrt{\gamma(\cal C)}}\right).
\end{equation}
Although \(\ETD\) and \(\gamma\) arose from different approaches to exact learning, they play closely related roles: both measure, in different ways, how much progress can be forced by a membership query. It is therefore natural to ask whether the two parameters are manifestations of the same underlying combinatorial phenomenon. We show that this is indeed the case after passing to natural fractional relaxations. We introduce fractional analogues \(\fgamma\) and \(\fetd\) of \(\gamma\) and \(\ETD\), respectively, satisfying
\[
    \frac{1}{\gamma(\cal C)}
    \leq
    \frac{1}{\fgamma(\cal C)}
    \qquad\text{and}\qquad
    \fetd(\cal C)\leq \etd(\cal C).
\]
Our main structural observation is that the distinction between the two measures disappears completely after fractionalization: \Cref{lemma:fgamma-fetd} shows that
\[
    \frac{1}{\fgamma(\cal C)}
    =
    \Theta\!\bigl(\fetd(\cal C)\bigr).
\]
This phenomenon is reminiscent of the relationship between block sensitivity and certificate complexity for Boolean functions. While \(\bs\) and \(\C\) can differ polynomially, Tal~\cite{Tal13} observed that their natural fractional relaxations coincide by linear programming duality:
$   \fbs(f,x)=\fc(f,x)$
for every input \(x\); see also~\cite{GSS16}. The analogy is particularly natural here because extended teaching dimension is closely related to certificate complexity, with a specifying set playing the role of a certificate. Lemma~\ref{lemma:fgamma-fetd} shows that an analogous collapse occurs for the splitting parameter and extended teaching dimension in exact learning. 
Moreover, this common fractional parameter has direct consequences for query complexity. 

\begin{theorem}\label{thm:qr-fetd}
    For every \(\cal C\subseteq \zone^N\),
    \[
        \Q(\cal C)
        =
        \Omega\!\left(\sqrt{\fetd(\cal C)}\right),
        \qquad
        \R(\cal C)
        =
        O\!\left(
            \frac{\fetd(\cal C)}
                 {\log(\fetd(\cal C)+1)}
            \log|\cal C|
        \right).
    \]
    Equivalently, by Lemma~\ref{lemma:fgamma-fetd},
    \[
        \Q(\cal C)
        =
        \Omega\!\left(
            \sqrt{\frac{1}{\fgamma(\cal C)}}
        \right),
        \qquad
        \R(\cal C)
        =
        O\!\left(
            \frac{1}
                 {\fgamma(\cal C)\log(1/\fgamma(\cal C)+1)}
            \log|\cal C|
        \right).
    \]
\end{theorem}
The possibility that \(\ETD\) might also yield quantum query lower bounds was raised by Arunachalam and de Wolf~\cite{arunachalam2017guest}. Theorem~\ref{thm:qr-fetd} confirms this intuition after fractionalization: the fractional extended teaching dimension captures, up to constant factors, the same combinatorial quantity as the splitting parameter underlying the Servedio-Gortler quantum lower bound. The theorem can also be viewed as simultaneously refining two classical bounds. On the lower-bound side, it strengthens Eq.~\eqref{eq:1/sqrt gamma lower bound} by replacing \(\gamma\) with the potentially smaller fractional parameter \(\fgamma\). On the upper-bound side, our argument refines the randomized simulation of Arunachalam et al.~\cite[Theorem~10]{arunachalam2021two},
\[
    \R(\cal C)
    =
    O\!\left(
        \frac{\Q(\cal C)^2}{\log \Q(\cal C)}
        \log|\cal C|
    \right),
\]
and may be viewed as a fractional analogue of the classical bounds of~\cite{hegedHus1995generalized,moshkov1982conditional} who showed
\[
    \D(\cal C)
    =
    O\!\left(
        \frac{\etd(\cal C)}
             {\log(\etd(\cal C)+1)}
        \log|\cal C|
    \right).
\]
There is one important distinction: the Heged\H{u}s bound is deterministic, whereas ours is randomized. In exchange, our bound replaces \(\etd(\cal C)\) by the potentially smaller and more refined quantity \(\fetd(\cal C)\). Our randomized upper bound is inspired by the framework of~\cite[Theorem~10]{arunachalam2021two}.

\paragraph{AI statement.} We used GPT-5.6 in the process of discovering most of our results, but notably the construction underlying our separation in \Cref{thm:rqseparation}. Our use of the model was iterative with a long thread of discussions. The model first provided us with a $\R = \Omega(\Q(\cal C)^2\log N)$ separation, which itself emerged only after substantial back-and-forth, with several intermediate candidates and repeated simplifications. We present the final simplified version as a warm-up separation in \Cref{sec:warmup}.  The final construction there is considerably simpler than the earlier versions considered during these discussions. 

We then provided several possible ingredients and directions, including standard quantum query separations, the idea of totalizing a partial construction using a cheat-sheet gadget and how to use that to construct a \emph{learning separation}. We then asked  it to search for candidate exact-learning classes inspired by these. GPT-5.6 produced an initial ``superquadratic" separation that ignored several nuances of exact learning and was considerably more complicated than necessary. After isolating the core mechanism, we suggested replacing the central gadget (i.e., Forrelation, which was its initial inspiration) by a hidden-shift problem (since it was using the Forrelation candidate as a $1$ vs $\Omega(n)$ quantum-classical query separation). GPT-5.6 then helped simplify the construction around this idea while preserving the same separation. The eventual construction was still fairly complicated, and the current version is after a substantial simplification. The resulting arguments were subsequently verified, refined, simplified and rewritten in final form by the authors and we take full responsibility for any mistakes.

\section{Optimal quantum-randomized separation}\label{sec:rqseparation}

\subsection{Warm up}
\label{sec:warmup}
Before proving our optimal quantum-randomized separation, we first give a simpler construction that already refutes the At{\i}ci-Servedio conjecture~\cite{atici2005improved} (\Cref{conj:aticiservedio}).
\[
    \R(\Cc)=O\!\left(\Q(\Cc)^2+\Q(\Cc)\log N\right).
\]
We first give the concept class witnessing this separation, together with a proof sketch of the quantum and classical bounds. For integers $\ell,m\ge 1$, consider an unknown matrix $A\in\F_2^{m\times \ell}$ and define the concept
\[
    c_A:\F_2^\ell\times\F_2^m\to\{0,1\},
    \qquad
    c_A(z,y)=\mathbf 1[Az=y].
\]
Let
\[
    \Cc_{\ell,m}:=\{c_A:A\in\F_2^{m\times\ell}\}.
\]
Thus, the learner is given membership-query access to an unknown matrix
$A\in\F_2^{m\times\ell}$ through queries of the form ``is $Az=y$?''

\paragraph{Quantum upper bound.}
A pseudocode for our algorithm witnessing the upper bound is in \Cref{alg:rq quantum upper bound}.
\begin{algorithm}[!ht]
\caption{Quantum exact learning algorithm for $\mathcal{C}_{\ell,m}$ from \Cref{sec:warmup}}
\label{alg:rq quantum upper bound}
\begin{algorithmic}[1]
\Require Query access to the oracle $O_A:(z,y)\mapsto\mathbb{I}[Az=y]$, where $A\in\mathbb{F}_2^{m\times \ell}$ is unknown.
\State For $i\in[m]$, let $r_i\in\mathbb{F}_2^\ell$ denote the $i$th row of $A$.
\For{$i=1,2,\ldots,m$}
    \State Define a phase oracle $O_i$ acting on $\ket{z}$ as follows:
    \State \hspace{1em} Compute the known prefix $(\langle r_1,z\rangle,\ldots,\langle r_{i-1},z\rangle)$. \label{line:rq known prefix}
    \State \hspace{1em} Grover-search to find the unique $y$ satisfying $O_A(z,y)=1$. \label{line:rq grover}
    \State \hspace{1em} Apply the phase $(-1)^{y_i}$ and uncompute. \label{line:rq phase}
    \State $r_i\gets\textnormal{Bernstein-Vazirani}(O_i)$. \label{line:rq bv}
\EndFor
\State \Return $A$.
\end{algorithmic}
\end{algorithm}
The algorithm is described below:
Write $r_1,\ldots,r_m\in\F_2^\ell$ for the rows of $A$. We learn them sequentially. To learn $r_i$, we use Grover search to coherently recover the $i$th bit of $Az$, and then Bernstein-Vazirani to recover the linear form $z\mapsto r_i\cdot z$. How we do this is described below.
Suppose $r_1,\ldots,r_{i-1}$ are already known. For any fixed $z\in\F_2^\ell$, the first $i-1$ coordinates of $Az$ are known, so the set of possible values of $Az$, namely $\{y\in\F_2^m:y_j=r_j\cdot z\text{ for every }j<i\}$, has size $2^{m-i+1}$.  The membership oracle itself serves as the Grover marking oracle, since $c_A(z,y)=1 \Longleftrightarrow y=Az$.
The reflection about the uniform superposition over the candidate values of $y$ is query-free: this superposition is prepared by a known unitary depending only on the previously learned rows $r_1,\ldots,r_{i-1}$. Among these $2^{m-i+1}$ candidates there is a unique correct value $y=Az$. Hence, using the membership oracle to test whether a candidate $y$ satisfies $c_A(z,y)=1$, exact Grover search (zero error) finds $Az$ using $O\!\left(2^{(m-i+1)/2}\right)$ membership queries.

Moreover, this search can be implemented coherently with $z$ in superposition. We can therefore compute the $i$th bit $(Az)_i=r_i\cdot z$ into an ancilla while leaving $z$ unchanged. Applying a phase conditioned on this ancilla and then uncomputing the search workspace implements the phase oracle
\[
    |z\rangle \longmapsto (-1)^{r_i\cdot z}|z\rangle .
\]
A Bernstein-Vazirani query to this phase oracle recovers the entire row $r_i$.
For the query upper bound, summing the costs over all $i$ gives
\begin{equation}\label{eq:warmup-ub}    
    \Q(\Cc_{\ell,m}) = O\!\left(\sum_{i=1}^m2^{(m-i+1)/2}\right)=    O(2^{m/2}).
\end{equation}
The matching lower bound follows by restricting to matrices of the form $A=(a,0,\ldots,0)$, which reduces learning to unstructured search over $a\in\F_2^m$. Hence $Q(\Cc_{\ell,m})=\Theta(2^{m/2}).$

\paragraph{Randomized lower bound.}
We apply Yao's minimax principle to the uniform distribution over $A\in\F_2^{m\times\ell}$.  Fix a deterministic decision tree making at most $T$ queries, after removing queries whose answers are already forced by the preceding transcript. At a node $v$, let $\mathcal A_v$ denote the set of matrices consistent
with the transcript so far, and define
\[
    D_v := \left\{z\in\F_2^\ell : Az=A'z \text{ for all }A,A'\in\mathcal A_v\right\}.
\]
Thus $D_v$ is the subspace of inputs on which the value of the unknown linear map has already been determined. Suppose a nonredundant query $(z,y)$ receives answer $1$.  Then every matrix surviving this answer satisfies $Az=y$, so $z$ belongs to the new determined subspace.  On the other hand, $z$ did not belong to $D_v$ before the query: otherwise $Az$ was already fixed by the transcript, and the answer to $(z,y)$ would have been forced.  Since $D_v$ is a subspace, every $1$-answer therefore increases $\dim D_v$ by at least one. As $D_v\subseteq\F_2^\ell$, every root-to-leaf transcript contains at most $\ell$ answers equal to $1$. 
Consequently, a depth-$T$ decision tree has at most
\begin{equation}\label{eq:leaves-ub}    
    \sum_{j=0}^{\ell}\binom{T}{j} \le (eT/\ell)^\ell
\end{equation}
reachable leaves. Each leaf outputs a single matrix and hence can be correct on at most one of the $2^{m\ell}$ possible targets. Therefore, success probability at least $2/3$ under the uniform distribution requires
\begin{equation}\label{eq:leaves-lb}    
    \sum_{j=0}^{\ell}\binom{T}{j}\ge \frac23\,2^{m\ell}.
\end{equation}
\Cref{eq:leaves-ub} and \Cref{eq:leaves-lb} imply $T=\Omega(2^m\ell)$. Yao's minimax principle therefore gives
\begin{equation}\label{eq:warmup-lb}
    \R(\Cc_{\ell,m})=\Omega(2^m\ell).
\end{equation}
 Finally, take $\ell=m$. Since the domain of $\Cc_{m,m}$ has size $N=2^{2m},$ we have $\log N=2m$. \Cref{eq:warmup-ub} and \Cref{eq:warmup-lb} thus give the separation of
\[
    \R(\Cc)=\Omega\!\left(\Q(\Cc)^2\log N\right).
\]

The lower bound presented above is closely related to the search-by-hyperplane-queries lower bound of Yun~\cite{yun2015generic}. Yun's result is stated over prime fields and for unrestricted affine-hyperplane queries, whereas our setting involves the restricted queries arising from the concept class above. We therefore included the short proof above for completeness.

\subsection{The concept class}
We now formally define the concept class that proves Theorem~\ref{thm:rqseparation}. An accompanying figure that describes concepts from this class is \Cref{fig:concept-class}.
\begin{defi}\label{def:class}
Let $t\ge4$ be a power of two. 
Fix a binary basis $e_1,\ldots,e_{6\log t}$ of $\mathbb F_{t^6}$. All inner products use coordinates in this basis and are taken over $\F_2$. Define truncation to the first $2\log t$ coordinates by the map
\[
\mathsf{trunc}:\mathbb F_{t^6}\longrightarrow \{0,1\}^{2\log t},\qquad
\mathsf{trunc}(w_1,\ldots,w_{6\log t})=(w_1,\ldots,w_{2\log t}).
\]
For an element $s\in \mathbb F_{t^6}$ and a monic polynomial
\[
    P(X)=X^t+\sum_{j=0}^{t-1}a_jX^j\in\mathbb F_{t^6}[X]
\]
of degree $t$, write $P_{\mathsf{trunc}}:=\mathsf{trunc}\circ P:\mathbb F_{t^6}\to\{0,1\}^{2\log t}$. For each pair $(P,s)$ as above, we define a concept $c_{P,s}$ with a \emph{block part} and a \emph{cheat sheet part}. We use a first coordinate $\mathsf{blk}$ or $\mathsf{sheet}$ to indicate which part is being queried, with respective domains
\[
\mathcal X_t^{\mathrm{blk}}=\mathbb F_{t^6}^{\,2}\times\{0,1\}^{2\log t},\qquad
\mathcal X_t^{\mathrm{sheet}}=\mathbb F_{t^6}\times\{0,\ldots,t-1\}\times\mathbb F_{t^6}.
\]

\noindent \textbf{Block part.} For $x,c\in\mathbb F_{t^6}$ and $z\in\{0,1\}^{2\log t}$, define
\[
c_{P,s}(\mathsf{blk},x,c,z)=\one[z=P_{\mathsf{trunc}}(c+xs)].
\]
For each fixed $(x,c)\in\mathbb F_{t^6}^{\,2}$, the corresponding block has $t^2=2^{2\log t}$ locations indexed by $z\in\{0,1\}^{2\log t}$, with a unique $1$ at $z=P_{\mathsf{trunc}}(c+xs)$.

\noindent \textbf{Cheat sheet part.} For $u,b\in\mathbb F_{t^6}$ and $j\in\{0,\ldots,t-1\}$, define
\[
c_{P,s}(\mathsf{sheet},u,j,b)=\one[u=s]\,\langle b,a_j\rangle.
\]
For each fixed $u\in\mathbb F_{t^6}$, the corresponding cheat sheet is the function $(j,b)\mapsto c_{P,s}(\mathsf{sheet},u,j,b)$. It is identically zero for $u\ne s$, while at $u=s$ its $j$th row is the linear function $b\mapsto\langle b,a_j\rangle$.

\noindent \textbf{Concept definition.} Let
\[
\mathcal X_t:=\{\mathsf{blk}\}\times\mathcal X_t^{\mathrm{blk}}\;\sqcup\;\{\mathsf{sheet}\}\times\mathcal X_t^{\mathrm{sheet}}.
\]
Thus an input to $c_{P,s}$ is either a block query $(\mathsf{blk},x,c,z)$, on which
\[
c_{P,s}(\mathsf{blk},x,c,z)=\one[z=P_{\mathsf{trunc}}(c+xs)],
\]
or a cheat sheet query $(\mathsf{sheet},u,j,b)$, on which
\[
c_{P,s}(\mathsf{sheet},u,j,b)=\one[u=s]\,\langle b,a_j\rangle.
\]
We suppress the tags $\mathsf{blk}$ and $\mathsf{sheet}$ when the intended part is clear from context. Hence
\[
\Cc_t=\left\{c_{P,s}: s\in\mathbb F_{t^6},\quad P(X)=X^t+\sum_{j=0}^{t-1}a_jX^j,\quad a_0,\ldots,a_{t-1}\in\mathbb F_{t^6}\right\}.
\]
The domain size is
\[
N_t:=|\mathcal X_t|=t^{14}+t^{13}=\Theta(t^{14}).
\]
The learner may make membership queries to either part in any adaptive order, and its goal is to identify the unknown concept, or equivalently, to recover both $P$ and $s$.
\end{defi}

\begin{figure}[H]
\centering
\resizebox{\linewidth}{!}{%
\begin{tikzpicture}[
  x=1cm,y=1cm,
  font=\normalsize,
  every node/.style={inner sep=2pt}
]
\definecolor{conceptblue}{RGB}{35,95,145}

\draw[rounded corners=5pt,black!25]
  (0,-.1) rectangle (7.8,6);
\draw[rounded corners=5pt,black!25]
  (8.3,-.1) rectangle (19.8,6);

\node[font=\bfseries] at (3.9,5.55) {Block part};
\node[font=\bfseries] at (14.05,5.55) {Cheat sheet part};


\node at (3.7,4.6) {$z\in \{0,1\}^{2\log t}$};
\fill[conceptblue!5] (1.3,1.45) rectangle (6.1,4.15);

\foreach \r/\k in {0/1,1/6,2/3,3/5,4/0,5/4} {
  \fill[conceptblue]
    ({1.3+.6*\k},{1.45+.45*\r})
    rectangle ({1.9+.6*\k},{1.9+.45*\r});
  \node[text=white]
    at ({1.6+.6*\k},{1.675+.45*\r}) {$1$};
}

\foreach \k in {0,...,8}
  \draw[black!30]
    ({1.3+.6*\k},1.45) -- ({1.3+.6*\k},4.15);

\foreach \r in {0,...,6}
  \draw[black!30]
    (1.3,{1.45+.45*\r}) -- (6.1,{1.45+.45*\r});

\draw[black!60] (1.3,1.45) rectangle (6.1,4.15);

\draw[decorate,decoration={brace,amplitude=4pt}]
  (1.05,1.45) -- (1.05,4.15);
\node at (.55,2.8) {$t^{12}$};

\draw[decorate,decoration={brace,mirror,amplitude=4pt}]
  (1.3,1.2) -- (6.1,1.2);
\node at (3.7,.85) {${t^2}$};

\draw[black!50] (6.1,3.025) -- (6.4,3.025);
\node[anchor=west] at (6.5,3.025) {$(x,c)$};

\node[anchor=east] at (3.95,.3) {$z=\mathsf{trunc}(P(c+x$};
\node[text=conceptblue] (blockslope) at (4,.29) {$s$};
\node[anchor=west] at (4,.29) {$))$};


\foreach \x in {9.0,18.3} {
  \draw[black!40,fill=black!3]
    (\x,1.45) rectangle ({\x+.95},4.15);
  \node at ({\x+.475},2.8) {\Large $0$};
  \node at ({\x+.475},.9) {$u\ne s$};
}

\node at (11.1,2.8) {$\cdots$};
\node at (17.65,2.8) {$\cdots$};

\node at (14.45,4.75) {$b\in\mathbb F_{t^6}$};

\fill[conceptblue!8] (12.5,1.45) rectangle (16.4,4.15);
\draw[conceptblue,thick]
  (12.5,1.45) rectangle (16.4,4.15);

\foreach \y in {2.125,2.8,3.475}
  \draw[conceptblue!40] (12.5,\y) -- (16.4,\y);

\node at (14.45,3.8125) {$\langle b,a_0\rangle$};
\node at (14.45,3.1375) {$\langle b,a_1\rangle$};
\node at (14.45,2.4625) {$\vdots$};
\node at (14.45,1.7875) {$\langle b,a_{t-1}\rangle$};

\draw[decorate,decoration={brace,mirror,amplitude=4pt}]
  (16.65,1.45) -- (16.65,4.15);
\node at (17.05,2.8) {$t$};

\draw[decorate,decoration={brace,mirror,amplitude=4pt}]
  (12.5,1.2) -- (16.4,1.2);
\node at (14.45,.85) {$t^6$};

\node[text=conceptblue] (cheatslope) at (14.45,.55) {$u=s$};

\draw[->,semithick,conceptblue]
  (blockslope.east)
  to[out=0,in=180]
  (cheatslope.west);

\end{tikzpicture}%
}
\caption{The two parts of $c_{P,s}$. Each block (represented by rows in the block part) has a unique marked address; blank entries are zero. The cheat sheet part consists of $t^6$ sheets indexed by $u\in \mathbb F_{t^6}$: sheet $s$ encodes the coefficients of $P$, and every other sheet is zero. The arrow shows that the $s$ part that is uncovered in the block part serves as the pointer for the row in the cheat sheet part.}
\label{fig:concept-class}
\end{figure}
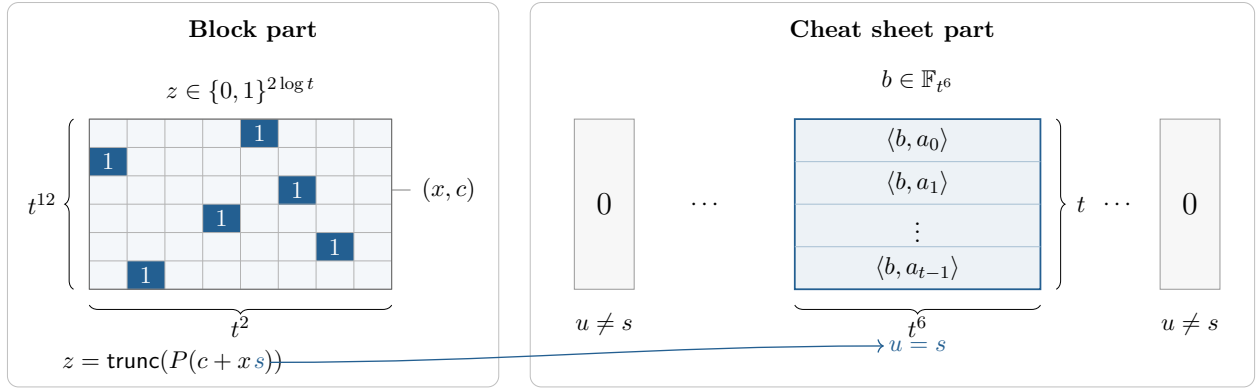

\subsubsection{Useful facts}
In this section we list useful facts about our concept class from \Cref{def:class} which will be useful in the analysis of the quantum upper bound and classical lower bound subsequently.
\begin{fact}\label{fact:truncation}
Every label in $\{0,1\}^{2\log t}$ has exactly $t^4$ preimages under $\mathsf{trunc}$.
Consequently, $\mathsf{trunc}$ sends a uniform element of $\mathbb F_{t^6}$ to a uniform
label in $\{0,1\}^{2\log t}$.
\end{fact}

\begin{proof}
Using the fixed binary basis of $\mathbb F_{t^6}$, we identify
 $  \mathbb F_{t^6} \cong \F_2^{6\log t}$, $    \{0,1\}^{2\log t} \cong \F_2^{2\log t}$ 
and $\mathsf{trunc}$ simply discards the last $4\log t$ coordinates.  Thus, for any fixed $z\in \{0,1\}^{2\log t}$, the first $2\log t$ coordinates are prescribed while the
remaining $4\log t$ coordinates may be chosen arbitrarily.  Hence
\[
    |\mathsf{trunc}^{-1}(z)|
    =
    2^{4\log t}=
   t^4.
\]
Since every $z\in \{0,1\}^{2\log t}$ has the same number of preimages, the image under
$\mathsf{trunc}$ of a uniformly random element of $\mathbb F_{t^6}$ is uniform on $\{0,1\}^{2\log t}$.
\end{proof}

\begin{fact}
\label{fact:frequency}
Let $P(X)=X^t+\sum_{j=0}^{t-1}a_jX^j\in\mathbb F_{t^6}[X]$ be monic of degree $t$. Then for every $z\in\{0,1\}^{2\log t}$, $|P_{\mathsf{trunc}}^{-1}(z)|\le t^5$, and $P_{\mathsf{trunc}}$ is nonconstant.
\end{fact}

\begin{proof}
Fix $z\in\{0,1\}^{2\log t}$. By Fact~\ref{fact:truncation}, there are exactly $t^4$ elements $w\in\mathbb F_{t^6}$ satisfying $\mathsf{trunc}(w)=z$. Therefore
\[
P_{\mathsf{trunc}}^{-1}(z)=\bigcup_{\mathsf{trunc}(w)=z}\{y\in\mathbb F_{t^6}:P(y)=w\}.
\]
For each such $w$, the polynomial $P(X)-w$ is monic (and hence a nonzero polynomial) of degree $t$, and hence has at most $t$ roots. Consequently,
\[
|P_{\mathsf{trunc}}^{-1}(z)|\le \sum_{w:\mathsf{trunc}(w)=z}|\{y\in\mathbb F_{t^6}:P(y)=w\}|\le t^4\cdot t=t^5.
\]
Since $t^5<t^6=|\mathbb F_{t^6}|$, no fiber of $P_{\mathsf{trunc}}$ is the entire domain, so $P_{\mathsf{trunc}}$ is nonconstant.
\end{proof}

\begin{fact}\label{fact:parameters}
The map $(P,s)\longmapsto c_{P,s}$ is injective. In particular, $|\mathcal C_t|=t^{6(t+1)}$.
\end{fact}

\begin{proof}
Suppose $c_{P,s}=c_{P',s'}$. Taking $x=0$ in the block part gives $P_{\mathsf{trunc}}=P'_{\mathsf{trunc}}$; write $f$ for this common function. Equality of the block parts then gives
\[
f(y)=f\bigl(y+x(s'-s)\bigr)\qquad\text{for every }y,x\in\mathbb F_{t^6}.
\]
If $s\ne s'$, then multiplication by $s'-s$ is a bijection of $\mathbb F_{t^6}$, so as $x$ ranges over $\mathbb F_{t^6}$, the quantity $x(s'-s)$ also ranges over all of $\mathbb F_{t^6}$. Hence $f$ is constant, contradicting \Cref{fact:frequency}. Therefore $s=s'$. At the cheat sheet address $u=s$, equality of the concepts gives $\langle b,a_j\rangle=\langle b,a'_j\rangle$ for every $b\in\mathbb F_{t^6}$ and every $j\in\{0,\ldots,t-1\}$, and hence $a_j=a'_j$ for every $j$. Thus $P=P'$.

Finally, since $P(X)=X^t+\sum_{j=0}^{t-1}a_jX^j$ is monic, each of its $t$ coefficients $a_0,\ldots,a_{t-1}$ can be chosen independently from $\mathbb F_{t^6}$, giving $(t^6)^t=t^{6t}$ choices for $P$. There are also $t^6$ choices for $s \in \mathbb{F}_{t^6}$, and therefore $|\mathcal C_t|=t^{6t}\cdot t^6=t^{6(t+1)}$.
\end{proof}

\begin{fact}
\label{fact:independence}
Suppose $a_0,\ldots,a_{t-1}$ are independent and uniformly chosen from $\mathbb F_{t^6}$. Then for any distinct $y_1,\ldots,y_r\in\mathbb F_{t^6}$ with $r\le t$, the random variables $P_{\mathsf{trunc}}(y_1),\ldots,P_{\mathsf{trunc}}(y_r)$ are independent and uniform in $\{0,1\}^{2\log t}$.
\end{fact}

\begin{proof}
Fix distinct $y_1,\ldots,y_r\in\mathbb F_{t^6}$ with $r\le t$, and write
\[
R(X):=\sum_{j=0}^{t-1}a_jX^j,
\qquad
P(X)=X^t+R(X).
\]
Consider the evaluation map
\[
\mathbb F_{t^6}^{\,t}\to\mathbb F_{t^6}^{\,r},
\qquad
(a_0,\ldots,a_{t-1})\mapsto\bigl(R(y_1),\ldots,R(y_r)\bigr).
\]
This map is surjective: for any prescribed values $v_1,\ldots,v_r\in\mathbb F_{t^6}$, polynomial interpolation gives a polynomial of degree at most $r-1<t$ taking value $v_i$ at $y_i$ for every $i$. Since the evaluation map is linear and surjective, all of its fibers have the same size. Hence, because $(a_0,\ldots,a_{t-1})$ is uniform in $\mathbb F_{t^6}^{\,t}$, the vector $(R(y_1),\ldots,R(y_r))$ is uniform in $\mathbb F_{t^6}^{\,r}$. Adding the fixed vector $(y_1^t,\ldots,y_r^t)$ shows that $(P(y_1),\ldots,P(y_r))$ is also uniform in $\mathbb F_{t^6}^{\,r}$. Finally, applying $\mathsf{trunc}$ coordinatewise and using \Cref{fact:truncation} shows that $\bigl(P_{\mathsf{trunc}}(y_1),\ldots,P_{\mathsf{trunc}}(y_r)\bigr)$ is uniform in $\bigl(\{0,1\}^{2\log t}\bigr)^r$, equivalently the random variables $P_{\mathsf{trunc}}(y_1),\ldots,P_{\mathsf{trunc}}(y_r)$ are independent and uniform in $\{0,1\}^{2\log t}$.
\end{proof}

\subsection{Quantum upper bound}
We now show that the concept class $\mathcal C_t$ can be learned with $O(t)$ quantum membership queries.
\begin{theorem}
\label{thm:rq-quantum-upper-bound}
Let $t\ge 4$ be a power of two, and let $\mathcal C_t$ be the concept class from \Cref{def:class}. Then
\[
\Q(\mathcal C_t)=O(t).
\]
\end{theorem}
\begin{proof}
Fix an unknown concept $c_{P,s}\in\mathcal C_t$. We first recover the hidden slope $s$ from the block part, and then learn $P$ from the cheat sheet part.

\begin{itemize}
\item Prepare the uniform superposition over all block coordinates,
\[
\frac{1}{t^7}\sum_{x,c\in\mathbb F_{t^6}}\sum_{z\in\{0,1\}^{2\log t}}\ket{x,c,z}.
\]
For every pair $(x,c)\in\mathbb F_{t^6}^2$, there is exactly one marked value of $z$, namely $z=P_{\mathsf{trunc}}(c+xs)$. Since there are $t^{12}$ choices of $(x,c)$ and $t^2$ choices of $z$, exactly a $1/t^2$ fraction of the basis states are marked. We may therefore use exact amplitude amplification to prepare
\begin{equation}
\ket{\psi_{P,s}}:=\frac{1}{t^6}\sum_{x,c\in\mathbb F_{t^6}}\ket{x,c,P_{\mathsf{trunc}}(c+xs)}
\label{eq:hidden-line-state}
\end{equation}
using $O(t)$ membership queries.

Apply the Hadamard transform $H^{\otimes 6\log t}$ to each of the first two registers of \Cref{eq:hidden-line-state}. This gives
\[
\frac{1}{t^{12}}\sum_{\alpha,\beta,x,c\in\mathbb F_{t^6}}(-1)^{\langle\alpha,x\rangle+\langle\beta,c\rangle}\ket{\alpha,\beta,P_{\mathsf{trunc}}(c+xs)}.
\]
Writing $y=c+xs$, this state becomes
\begin{equation}
\frac{1}{t^{12}}\sum_{\alpha,\beta,y\in\mathbb F_{t^6}}(-1)^{\langle\beta,y\rangle}\left(\sum_{x\in\mathbb F_{t^6}}(-1)^{\langle\alpha,x\rangle+\langle\beta,xs\rangle}\right)\ket{\alpha,\beta,P_{\mathsf{trunc}}(y)}.
\label{eq:fourier-state}
\end{equation}
For fixed $s\in\mathbb F_{t^6}$, multiplication by $s$ is an $\mathbb F_2$-linear map on $\mathbb F_{t^6}$. Let $M_s\in\mathbb F_2^{(6\log t)\times(6\log t)}$ denote its matrix with respect to our fixed binary basis, so that $xs=M_sx$. Hence $\langle\beta,xs\rangle=\langle\beta,M_sx\rangle=\langle M_s^\top\beta,x\rangle$.
The inner sum in \Cref{eq:fourier-state} is therefore equal to
\[
\sum_{x\in\mathbb F_{t^6}}(-1)^{\langle\alpha+M_s^\top\beta,x\rangle}.
\]
By orthogonality of characters, this sum equals $t^6$ when $\alpha=M_s^\top\beta$ and is zero otherwise. Hence the state in \Cref{eq:fourier-state} simplifies to
\[
\frac{1}{t^6}\sum_{\beta,y\in\mathbb F_{t^6}}(-1)^{\langle\beta,y\rangle}\ket{M_s^\top\beta,\beta,P_{\mathsf{trunc}}(y)}.
\]
Measuring the first two registers therefore returns a pair $(\alpha,\beta)$ satisfying $\alpha=M_s^\top\beta$.
If $\beta\ne 0$, then the relation $\alpha=M_s^\top\beta$ uniquely determines $s$. Indeed, the $i$th coordinate of $\alpha$ is
\[
\alpha_i=\langle\beta,e_is\rangle,\qquad i=1,\ldots,6\log t,
\]
since the $i$th column of $M_s$ is the coordinate vector of $e_is$. These are $6\log t$ linear equations over $\mathbb F_2$ in the $6\log t$ coordinate bits of $s$.

To see uniqueness, suppose $s'$ also satisfies $\alpha=M_{s'}^\top\beta$. Then $M_{s-s'}^\top\beta=0$, and hence $\langle\beta,x(s-s')\rangle=0$ for every $x\in\mathbb F_{t^6}$. If $s\ne s'$, multiplication by $s-s'$ is a bijection of $\mathbb F_{t^6}$, so this implies $\langle\beta,w\rangle=0$ for every $w\in\mathbb F_{t^6}$, and therefore $\beta=0$, a contradiction. Thus $s$ is uniquely determined whenever $\beta\ne 0$.

It remains to bound the probability that $\beta=0$. Since $\alpha=M_s^\top\beta$, this also forces $\alpha=0$. For each $z\in\{0,1\}^{2\log t}$, the amplitude of $\ket{0,0,z}$ is
\[
\frac{1}{t^6}\left|\{y\in\mathbb F_{t^6}:P_{\mathsf{trunc}}(y)=z\}\right|=\frac{|P_{\mathsf{trunc}}^{-1}(z)|}{t^6}.
\]
Therefore
\begin{equation}\label{eq:prbeta}
    \Pr[\beta=0]=\frac{1}{t^{12}}\sum_{z\in\{0,1\}^{2\log t}}|P_{\mathsf{trunc}}^{-1}(z)|^2.
\end{equation}

By \Cref{fact:frequency}, $\max_z |P_{\mathsf{trunc}}^{-1}(z)|\le t^5$. Moreover, $\sum_z |P_{\mathsf{trunc}}^{-1}(z)|=|\mathbb F_{t^6}|=t^6$ since the fibers of $P_{\mathsf{trunc}}$ partition $\mathbb F_{t^6}$.
Hence \Cref{eq:prbeta} gives
\[
\Pr[\beta=0]\le \frac{1}{t^{12}}\left(\max_z |P_{\mathsf{trunc}}^{-1}(z)|\right)\sum_z |P_{\mathsf{trunc}}^{-1}(z)|\le \frac{t^{11}}{t^{12}} = \frac1t.
\]
Thus the hidden slope $s$ is recovered with probability at least $1-1/t$.

\item Finally, once $s$ is known, the learner knows the unique active cheat sheet (recall \Cref{def:class}). Also recall that for each $j\in\{0,\ldots,t-1\}$, the $j$th row of the active sheet $u=s$ is the Boolean function
\[
b\longmapsto c_{P,s}(\mathsf{sheet},s,j,b)=\langle b,a_j\rangle,
\]
where $a_j\in\mathbb F_{t^6}\cong\F_2^{6\log t}$ is the $j$th coefficient of $P$. Thus each row is a Bernstein--Vazirani instance with hidden string $a_j$. One quantum membership query therefore recovers $a_j$ exactly, and repeating this for $j=0,\ldots,t-1$ recovers all coefficients of $P$ using $t$ additional queries.
\end{itemize}

Combining the query cost of $t$ from the last bullet to learn $P$ with the $O(t)$ queries used in the previous bullets to recover $s$, the total query complexity is $O(t)+t=O(t)$. The only possible failure occurs when the Fourier sample has $\beta=0$, which happens with probability at most $1/t$. Hence the learner succeeds with bounded error, proving $\Q(\mathcal C_t)=O(t)$.
\end{proof}

\subsection{Randomized lower bound}

We now prove our randomized lower bound for $\mathcal C_t$.

\begin{theorem}
\label{thm:rq-randomized-lower-bound}
Let $t\ge 4$ be a power of two, and let $\mathcal C_t$ be the concept class from \Cref{def:class}.~Then
\[
\R(\mathcal C_t)=\Omega(t^3).
\]
\end{theorem}
We prove the lower bound using Yao's minimax principle~\cite{yao1977probabilistic}. Choose $s,a_0,\ldots,a_{t-1}$ independently and uniformly from $\mathbb F_{t^6}$, where $s$ is the hidden slope and $a_0,\ldots,a_{t-1}$ are the coefficients of the unknown polynomial $P$, and let $\mu$ be the resulting distribution on $\mathcal C_t$. We will show that every deterministic learner making $D\ll t^3$ queries has success probability less than $2/3$ under the input distribution $\mu$.

We first describe the proof strategy. The idea is to compare the learner's real execution under $\mu$ with a sequence of simpler fictitious executions. At each step, we modify how some queries are answered, and show that with high probability the learner cannot distinguish the modified execution from the previous one. After three such modifications, we arrive at an execution whose transcript is independent of the hidden slope $s$, and is therefore easy to analyze. Since the real execution is close to this final fictitious one, any learner that succeeds with high probability in the real experiment would also have to succeed in the fictitious experiment with noticeable probability, which we will show is impossible with $D\ll t^3$ queries.

\begin{enumerate}
    \item First, we modify the execution by answering every cheat-sheet query with $0$, while leaving all block answers unchanged. The real and modified executions agree until the learner first queries the active sheet $u=s$. Thus, if the learner succeeds in the real execution, then either it queries the active sheet in both executions, or the two executions agree throughout and its final output has the correct slope $s$. It therefore suffices to bound the probability that, in the modified execution, the learner either queries the active sheet or outputs a concept with slope $s$.
    \item Second, we replace $P_{\mathrm{trunc}}$ by a uniformly random function $g:\mathbb F_{t^6}\to\{0,1\}^{2\log t}$. This makes the labels at distinct hidden inputs independent, removing the dependencies imposed by the polynomial. By \Cref{fact:independence}, $P_{\mathrm{trunc}}$ is $t$-wise independent, but the learner may query more than $t$ distinct hidden inputs, so this alone does not justify the replacement. We also use the fact that each block query only tests whether the label equals one specified address among $t^2$ possibilities. We prove that, when $D\ll t^3$, this replacement changes the probability of the event from the previous bullet, namely that the learner queries the active sheet or outputs the correct slope, by only a small amount. Since the previous bullet bounds the learner's success probability by the probability of this event, it suffices to bound its probability in the new experiment, allowing for this small error.
    \item Finally, we assign each block $(x,c)$ an independent uniformly random marked address from $\{0,1\}^{2\log t}$, which remains fixed throughout the execution. To compare this with the previous experiment, observe that a random function $g$ already assigns independent uniformly random marked addresses to blocks with distinct hidden inputs $c+xs$. We can therefore arrange for the two experiments to give exactly the same answers unless two distinct queried blocks have the same hidden input. In the new experiment, the learner's queries are independent of $s$, and each pair of distinct blocks has the same hidden input for at most one slope. A union bound thus bounds the probability of such a collision by $\binom{D}{2}/t^6$, which also bounds the change in the probability of the event tracked in the previous bullets.

    In this final experiment, the learner's entire transcript is independent of $s$. Its cheat-sheet queries and final output therefore give at most $D+1$ guesses for a uniformly random slope in $\mathbb F_{t^6}$. Hence the probability that it queries the active sheet or outputs the correct slope is at most $(D+1)/t^6$.
\end{enumerate} 
Combining the three steps, the learner's success probability in the real execution is at most $$(D+1)/t^6+\binom{D}{2}/t^6+\binom{D}{t}(2/t^2)^t.
$$ Taking $D=\lfloor t^3/10\rfloor$, the first term is $O(t^{-3})$, the second is at most $1/200$, and the third is at most $(2eD/t^3)^t\leq(e/5)^t$. Thus the success probability is less than $2/3$ for all sufficiently large $t$. By Yao's minimax principle, this proves $R(\mathcal C_t)=\Omega(t^3)$.

We now formalize the hybrid argument. Fix a deterministic learner making at most $D$ queries. We run this same learner in the real experiment and in each of the hybrids defined below. In each experiment, $T$ denotes the transcript produced in that experiment, including the queries, answers, and final output; its distribution therefore depends on the experiment. We may assume that the learner always outputs a concept in $\mathcal C_t$, since replacing an invalid output by an arbitrary fixed concept in $\mathcal C_t$ cannot decrease its success probability. By \Cref{fact:parameters}, each concept has a unique parameter pair $(P,s)$. We can therefore read a slope guess from the learner's output, which we denote by $\widehat s(T)$. In the real experiment, successful identification of the target concept $c_{P,s}$ necessarily implies $\widehat s(T)=s$. In the hybrids, we will track whether this slope guess is correct or the learner queries the sheet indexed by $s$.

\subsubsection{Hybrid 1}
Sample $(P,s)$ from $\mu$ exactly as in the original learning problem. Answer each block query $(x,c,z)$ by $\one[P_{\mathsf{trunc}}(c+xs)=z]$, but answer every cheat-sheet query by $0$. Write $\Pr_{\mathrm{H1}}$ for probability in this experiment, and let $T$ denote its transcript. Define
\begin{equation}\label{eq:a}
    A:=\{\text{the learner queries a cheat sheet with address }u=s\}\cup\{\widehat s(T)=s\}.
\end{equation}

Thus $A$ occurs if the learner queries the active sheet or its final output has the correct slope. Querying the active sheet does not require the learner to recognize that it is active. Let $\Pr_\mu[\mathrm{success}]$ denote the probability that the fixed deterministic learner outputs the target concept $c_{P,s}$ in the real experiment, where $(P,s)\sim\mu$.

\begin{lemma}\label{lem:zero-sheet}
Let $A$ be as in \Cref{eq:a}. For every deterministic learner, $\Pr_\mu[\mathrm{success}]\leq\Pr_{\mathrm{H1}}[A]$.
\end{lemma}

\begin{proof}
Fix a pair $(P,s)$ and run the deterministic learner in both the real experiment and Hybrid~1 with this same pair. Block queries receive the same answers in both experiments, as do queries to inactive cheat sheets. Thus, until the learner first queries the active sheet, it receives the same answers and makes the same subsequent queries in both executions.

If the learner queries the active sheet in the real experiment, it makes the same query in Hybrid~1, since the two executions agree up to that query. Hence the first event in the definition of $A$ from \Cref{eq:a} occurs in Hybrid~1. Otherwise, the learner never queries the active sheet in the real experiment, so the two executions agree throughout and have the same final output. If the learner succeeds in the real experiment, this output has slope $s$. Hence $\widehat s(T)=s$ in Hybrid~1, so the second event in the definition of $A$ occurs.
 In either case, whenever the learner succeeds in the real experiment, the event $A$ occurs in Hybrid~1. Taking probabilities over $(P,s) \sim \mu$,
\[
    \Pr_\mu[\mathrm{success}] \le \Pr_{\mathrm{H1}}[A],
\]
proving the lemma.
\end{proof}

\subsubsection{Hybrid 2}
Choose $s\in\mathbb F_{t^6}$ uniformly at random and, independently, choose a uniformly random function $g:\mathbb F_{t^6}\to\{0,1\}^{2\log t}$. Thus the values $g(y)$ at distinct inputs are independent and uniform in $\{0,1\}^{2\log t}$. Answer each block query $(x,c,z)$ by $\one[g(c+xs)=z]$, and continue to answer every cheat-sheet query by $0$. Write $\Pr_{\mathrm{H2}}$ for probability in this experiment.

The following lemma will allow us to compare Hybrid~1, where the marked address in block $(x,c)$ is $P_{\mathsf{trunc}}(c+xs)$, with Hybrid~2, where it is $g(c+xs)$. By \Cref{fact:independence}, the values of $\Ptrunc$ on at most $t$ distinct inputs are independent and uniform in $\{0,1\}^{2\log t}$. The learner may of course query more than $t$ distinct inputs, so this fact alone does not justify replacing $\Ptrunc$ by a fully random function. The following lemma shows that such a replacement nevertheless changes the transcript distribution only slightly, because each query only tests whether a function value equals one specified address.

\begin{lemma}\label{lem:equality}
Let $X,Z$ be finite sets with $|Z|=\ell$, and let $k\geq1$. Let $F:X\to Z$ be a random $k$-wise independent function, and let $G:X\to Z$ be a uniformly random function. Then any deterministic adaptive algorithm making at most $D$ equality queries of the form $\mathbf 1[F(x)=z]$ in the first experiment and $\mathbf 1[G(x)=z]$ in the second induces transcript distributions whose total variation distance is at most
\[
\binom{D}{k}\left(\frac{2}{\ell}\right)^k.
\]
\end{lemma}
\begin{proof}
Let $\Delta_k(D)$ denote the supremum, over all deterministic decision trees of depth at most $D$, of the total variation distance between the transcript distributions obtained when the queried function is $k$-wise independent and when it is fully random. We prove the desired bound by induction on $k$. Clearly $\Delta_0(D)\leq 1$, while $\Delta_k(0)=0$ for every $k$, since a depth-zero decision tree makes no queries.

Fix a deterministic decision tree of depth at most $D$ and an event $E$ of transcripts. Label each leaf by $1$ if its transcript lies in $E$, and by $0$ otherwise, and let $h(f)$ be the resulting output on a function $f:X\to Z$. Follow the path obtained by answering $0$ to every equality query, and let $(x_1,z_1),\ldots,(x_d,z_d)$, where $d\leq D$, be the queries encountered on this path. Let $h_0$ be the label of the terminal leaf. For each $j$ and $b\in\{0,1\}$, let $h_j^b:Z^X\to\{0,1\}$ denote the Boolean function that maps an oracle $f:X\to Z$ to the final leaf label obtained by continuing the computation from the $b$-child of the $j$th node with oracle $f$. Then
\begin{equation}\label{eq:h}
h(f)
=
h_0+\sum_{j=1}^d
\mathbf 1[f(x_j)=z_j]
\bigl(h_j^1(f)-h_j^0(f)\bigr).
\end{equation}
This identity follows by repeatedly expanding the computation along the all-zero path, starting from
\[
h(f)=h_1^0(f)+\mathbf 1[f(x_1)=z_1]\bigl(h_1^1(f)-h_1^0(f)\bigr).
\]

Now suppose $k\geq1$. Since both $F$ and $G$ are uniform at every individual input,
\[
\Pr[F(x_j)=z_j]=\Pr[G(x_j)=z_j]=\frac1\ell.
\]
By the assumptions in the lemma, conditioned on $F(x_j)=z_j$, the values of $F$ on any $k-1$ distinct inputs different from $x_j$ remain independent and uniform, while conditioned on $G(x_j)=z_j$, the values of $G$ on the remaining inputs remain fully independent and uniform. Any subsequent query to $x_j$ has a predetermined answer under this conditioning and may therefore be removed. Thus, for each $b\in\{0,1\}$,
\begin{equation}\label{eq:hexp}
\left|
\mathbb E\!\left[h_j^b(F)\mid F(x_j)=z_j\right]
-
\mathbb E\!\left[h_j^b(G)\mid G(x_j)=z_j\right]
\right|
\leq
\Delta_{k-1}(D-j).
\end{equation}
Importantly, we condition only on the single event $F(x_j)=z_j$ (respectively $G(x_j)=z_j$), and not on the preceding zero answers along the all-zero path. This is what allows us to retain $(k-1)$-wise independence after conditioning. Taking expectations in \Cref{eq:h} gives
\[
\mathbb E[h(F)]
=
h_0+\frac1\ell\sum_{j=1}^d
\mathbb E\!\left[
h_j^1(F)-h_j^0(F)
\,\middle|\,
F(x_j)=z_j
\right],
\]
and similarly
\[
\mathbb E[h(G)]
=
h_0+\frac1\ell\sum_{j=1}^d
\mathbb E\!\left[
h_j^1(G)-h_j^0(G)
\,\middle|\,
G(x_j)=z_j
\right].
\]
Subtracting these expressions, applying the triangle inequality, and using \Cref{eq:hexp} for $b=0,1$ gives
\[
\left|\mathbb E[h(F)]-\mathbb E[h(G)]\right|
\leq
\frac{2}{\ell}\sum_{j=1}^d \Delta_{k-1}(D-j).
\]
Since $h$ is the indicator of the chosen transcript event $E$, we have $\mathbb E[h(F)]=\Pr_F[E]$ and $\mathbb E[h(G)]=\Pr_G[E]$. Since $d\leq D$, taking the supremum over all deterministic decision trees of depth at most $D$ and all transcript events $E$ therefore gives
\[
\Delta_k(D)
\leq
\frac{2}{\ell}\sum_{j=1}^D \Delta_{k-1}(D-j).
\]
By induction,
\[
\Delta_k(D)\leq \left(\frac{2}{\ell}\right)^k\sum_{j=1}^D\binom{D-j}{k-1}
=\binom{D}{k}\left(\frac{2}{\ell}\right)^k,
\]
where the last equality uses $\sum_{r=k-1}^{D-1}\binom{r}{k-1}=\binom{D}{k}$.
Since total variation distance is the maximum difference in the probability of an event, this proves the lemma. 
\end{proof}
We now show that the transcript distributions in Hybrid 1 and Hybrid 2 are close in total variation distance.
\begin{lemma}\label{lem:replace-poly}
The joint distributions of $(s,T)$ in Hybrid~1 and Hybrid~2 have total variational distance at most 
\[
    \varepsilon_D
    :=
    \binom Dt\left(\frac2{t^2}\right)^t.
\]
Consequently,
$    \Pr_{\mathrm{H1}}[A]
    \le
    \Pr_{\mathrm{H2}}[A]+\varepsilon_D.$
\end{lemma}

\begin{proof}
Fix a value of $s\in \mathbb F_{t^6}$.  In Hybrid~1, every block query
$(x,c,z)$ is an equality query to the function $P_{\mathsf{trunc}}$ at the input
$    y=c+xs$, i.e., 
$    \one[P_{\mathsf{trunc}}(y)=z].$ 
By Fact~\ref{fact:independence}, the values of $P_{\mathsf{trunc}}$  on any set of at most \(t\) distinct inputs are independent and uniform in $\{0,1\}^{2\log t}$.  In Hybrid~2,
these labels are instead generated by a fully random function
$g:\mathbb F_{t^6}\to \{0,1\}^{2\log t}$.  Therefore, by Lemma~\ref{lem:equality}, conditioned on this
fixed value of $s$, the two transcript distributions differ in total
variation distance by at most
\[
    \binom Dt\left(\frac2{t^2}\right)^t
    =
    \varepsilon_D.
\]
Since $s$ is uniform in both hybrids, averaging over $s$ gives the same
bound for the joint distributions of $(s,T)$.  Since $A$ is an event
determined by $(s,T)$, we have 
$  \Pr_{\mathrm{H1}}[A]
    \le
    \Pr_{\mathrm{H2}}[A]+\varepsilon_D.$ 
\end{proof}

\subsubsection{Hybrid 3}

In Hybrid 3, assign to each block $(x,c)$ an independent uniformly random marked address $U_{x,c}\in\{0,1\}^{2\log t}$, which is reused on repeated queries to the same block. Thus, a block query $(x,c,z)$ is answered by $\mathbf{1}[U_{x,c}=z]$.
As in the previous hybrids, every cheat-sheet query is answered by $0$. Write $\Pr_{\mathrm{H3}}$ for probability in this experiment.

The only difference from Hybrid 2 is that there the marked address of block $(x,c)$ is $g(c+xs)$, where $g$ is a uniformly random function. Hence two distinct blocks receive the same marked address whenever they correspond to the same hidden input $c+xs=c'+x's$, whereas in Hybrid 3 distinct blocks are assigned independent marked addresses.

Recall that $T$ denotes the transcript produced by the learner in the hybrid under consideration, including its queries, the corresponding answers, and its final output.
\begin{lemma}\label{lem:replace-blocks}
The joint distributions of $(s,T)$ in Hybrid~2 and Hybrid~3 differ in total variation distance by at most
\[
    {\binom D2}/{t^6}.
\]
Consequently, for $A$ as in \Cref{eq:a},
$    \Pr_{\mathrm{H2}}[A]
    \le
    \Pr_{\mathrm{H3}}[A]
    +
    {\binom D2}/{t^6}.$
\end{lemma}

\begin{proof}
We couple Hybrid 2 and Hybrid 3 as follows. Sample the same uniformly random slope $s$ in both experiments. Whenever the learner queries a block $(x,c)$ for the first time, let $y=c+xs$. If no previously queried block $(x',c')$ satisfies $c'+x's=y$, sample a fresh uniformly random marked address $U\in\{0,1\}^{2\log t}$, set $g(y):=U$ in Hybrid 2, and set $U_{x,c}:=U$ in Hybrid 3. Repeated queries to the same block use the previously assigned marked address in both hybrids. Thus the two transcripts remain identical unless the learner queries
two distinct blocks $(x,c)\neq(x',c')$ satisfying
\begin{equation}\label{eq:collision}
    c+xs=c'+x's.
\end{equation}
Let $C$ denote the event that such a pair occurs among the queried blocks. Let $T_{\mathrm{H2}}$ and $T_{\mathrm{H3}}$ denote the transcripts in Hybrid 2 and Hybrid 3, respectively, under the coupling above, and set $X:=(s,T_{\mathrm{H2}})$ and $Y:=(s,T_{\mathrm{H3}})$. By construction, the marginal distributions of $X$ and $Y$ are precisely $(s,T)_{\mathrm{H2}}$ and $(s,T)_{\mathrm{H3}}$, respectively. By the standard coupling inequality, the total variation distance between $X$ and $Y$ is bounded by
\[
d_{\mathrm{TV}}(X,Y)\le \Pr[X\neq Y].
\]
Under our coupling, $X=Y$ whenever $C$ does not occur. Hence
\[
d_{\mathrm{TV}}\bigl((s,T)_{\mathrm{H2}},(s,T)_{\mathrm{H3}}\bigr)
\le \Pr_{\mathrm{H3}}[C].
\]
It remains to bound $\Pr_{\mathrm{H3}}[C]$. Condition on the complete transcript $T$ in Hybrid 3. Since the unique marked address assigned to each block is chosen independently of $s$, and every cheat-sheet query is answered by $0$, the transcript is independent of $s$. Hence, conditioned on $T$, the slope $s$ remains uniform in $\mathbb{F}_{t^6}$. The transcript contains at most $D$ distinct queried blocks. Fix two distinct queried blocks $(x,c)$ and $(x',c')$. If $x=x'$, then $c+xs=c'+x's$ implies $c=c'$, and hence $(x,c)=(x',c')$, contradicting that the two blocks are distinct. If $x\neq x'$, then $c+xs=c'+x's$ is equivalent to $s=(c'-c)/(x-x')$. Thus this equality holds for exactly one value of $s\in\mathbb{F}_{t^6}$, and therefore
\[
\Pr_{\mathrm{H3}}\!\left[c+xs=c'+x's\mid T\right]\le \frac{1}{t^6}.
\]
A union bound over at most $\binom{D}{2}$ pairs of queried blocks gives $\Pr_{\mathrm{H3}}[C\mid T]\le \frac{\binom{D}{2}}{t^6}$. Since this holds for every transcript $T$, averaging over $T$ yields $\Pr_{\mathrm{H3}}[C]\le \binom{D}{2}/t^6$, proving the lemma.
\end{proof}

\begin{lemma}\label{lem:guess}
In Hybrid 3, $\Pr_{\mathrm{H3}}[A]\le \frac{D+1}{t^6}$.
\end{lemma}

\begin{proof}
Condition on the complete transcript $T$. In Hybrid 3, the marked address of each block is chosen independently of $s$, and every cheat-sheet query is answered by $0$. Hence $T$ is independent of $s$, so conditioned on $T$, the slope $s$ remains uniform in $\mathbb F_{t^6}$.

The transcript determines the at most $D$ cheat-sheet addresses queried by the learner, as well as its final slope guess $\hat s(T)$. Therefore, conditioned on $T$, the event $A$ can occur only if $s$ equals one of at most $D+1$ specified elements of $\mathbb F_{t^6}$. Thus
\[
\Pr_{\mathrm{H3}}[A\mid T]\le \frac{D+1}{t^6}.
\]
Since this bound holds for every transcript $T$,
\[
\Pr_{\mathrm{H3}}[A]
=
\sum_T \Pr_{\mathrm{H3}}[T]\Pr_{\mathrm{H3}}[A\mid T]
\le
\frac{D+1}{t^6},
\]
proving the lemma.
\end{proof}

\subsubsection{Putting hybrids together}
We now show how to put the hybrids together to obtain our randomized lower bound in \Cref{thm:rq-randomized-lower-bound}.

\begin{proof}[Proof of \Cref{thm:rq-randomized-lower-bound}]
Recall the definition of $A$ from \Cref{eq:a} and that $\mu$ is the distribution obtained by choosing $s,a_0,\ldots,a_{t-1}$ independently and uniformly from $\mathbb F_{t^6}$. We use Yao's minimax principle~\cite{yao1977probabilistic} and show that the success probability of any deterministic learner that makes less than $t^3/10$ queries is small under the input distribution $\mu$.
Combining \Cref{lem:zero-sheet,lem:replace-poly,lem:replace-blocks,lem:guess} gives
\begin{align}
\label{eq:puttinghybridgstogether}
    \Pr_\mu[\mathrm{success}]
    &\le
    \Pr_{\mathrm{H1}}[A] \le
    \Pr_{\mathrm{H2}}[A]+\varepsilon_D \le
    \Pr_{\mathrm{H3}}[A]
    +\frac{\binom D2}{t^6}
    +\varepsilon_D \le
    \frac{D+1}{t^6}
    +\frac{\binom D2}{t^6}
    +\varepsilon_D,
\end{align}
where
\[
    \varepsilon_D
    =
    \binom Dt\left(\frac2{t^2}\right)^t.
\]
Now, let 
$    D={t^3}/{10}.$
Thus
\[
    \binom Dt\left(\frac2{t^2}\right)^t
    \le
    \left(\frac{eD}{t}\right)^t
    \left(\frac2{t^2}\right)^t
    =
    \left(\frac{2eD}{t^3}\right)^t
    \le
    \left(\frac e5\right)^t,
\]
while
\[
    \frac{\binom D2+D+1}{t^6}
    \le
    \frac1{200}+O(t^{-3}).
\]
Hence, for all sufficiently large $t$, the right-hand side of \Cref{eq:puttinghybridgstogether} is strictly smaller than $2/3$.

Thus, for all sufficiently large $t$, every deterministic learner making at most $t^3/10$ queries has success probability less than $2/3$ under $\mu$. By Yao's minimax principle, $\R(\mathcal C_t)=\Omega(t^3)$.
\end{proof}

\subsection{Putting it all together}
\rqseparation*

\begin{proof}
We have $|\mathcal C_t|=t^{6(t+1)}$ from \Cref{fact:parameters} and $N_t=\Theta(t^{14})$ from \Cref{def:class}. The standard quantum lower bound for exact learning~\cite[Theorem~10]{servedio2004equivalences} implies
\[
  \Q(\Cc_t)=\Omega\!\left(\frac{\log|\Cc_t|}{\log N_t}\right)
  =\Omega\!\left(\frac{6t\log t}{\log N_t}\right)
  =\Omega(t).
\]
Hence $\log N_t=\Theta(\log \Q(\Cc_t))$. Together with the quantum upper bound from \Cref{thm:rq-quantum-upper-bound}, this gives $\Q(\Cc_t)=\Theta(t)$. Applying the randomized lower bound from \Cref{thm:rq-randomized-lower-bound}, we obtain
\[
  \R(\Cc_t)=\Omega(t^3)
  =\Omega\!\left(\Q(\Cc_t)^3
               \frac{\log N_t}{\log \Q(\Cc_t)}\right),
\]
as required.
\end{proof}

\section{Optimal quantum-deterministic separation}
\label{sec:dqseparation}
In this section, we prove Theorem~\ref{thm:dqseparation}.  Like before we first give the concept class, then describe the classical and quantum bounds. 

\subsection{The concept class}
The construction is based on combining two standard sources of quantum speedup, Grover search and Bernstein-Vazirani. We partition the domain into $q^2$ blocks and choose one hidden block $b$. On every block $b'\neq b$, the concept is identically zero, while on block $b$ it is given by $(x,y)\mapsto x^\top A y$ for an unknown nonzero matrix $A\in\mathbb F_2^{q\times q}$. We now formalize this construction.

\begin{defi}\label{def:bilinear-class}
Let $q\ge2$ be an integer and let
$\mathcal X_q=[q^2]\times\F_2^q\times\F_2^q$.
For $b\in[q^2]$ and nonzero $A\in\F_2^{q\times q}$, define
$c_{b,A}:\mathcal X_q\to\{0,1\}$ by
$$c_{b,A}(b',x,y)=\one[b'=b]\,x^{\mathsf T}Ay.$$
Let $c_0$ be the identically zero concept, and define
\[
  \Cc_q=\Big\{c_{b,A}:\mathcal X_q\rightarrow \{0,1\}:b\in[q^2],\ A\in\F_2^{q\times q}\setminus\{0\}\Big\}
        \cup\{c_0\}.
\]
\end{defi}

\begin{figure}[ht]
\centering
\begin{tikzpicture}[x=1cm,y=1cm,font=\small]
    \draw (0,0) rectangle (2.4,2.4);
    \node at (1.2,1.2) {\Large $0$};
    \node at (1.2,2.75) {$b'\ne b$};

    \node at (3,1.2) {$\cdots$};

    \draw[fill=black!5] (3.6,0) rectangle (6,2.4);
    \node at (4.8,1.2) {$x^{\mathsf T}Ay$};
    \node at (4.8,2.75) {$b'=b$};

    \node at (6.6,1.2) {$\cdots$};

    \draw (7.2,0) rectangle (9.6,2.4);
    \node at (8.4,1.2) {\Large $0$};
    \node at (8.4,2.75) {$b'\ne b$};

    \draw[decorate,decoration={brace,amplitude=4pt}]
        (-0.2,0) -- (-0.2,2.4)
        node[midway,left=6pt] {$2^q$};
    \draw[decorate,decoration={brace,amplitude=4pt}]
        (0,3.05) -- (2.4,3.05)
        node[midway,above=6pt] {$2^q$};

    \draw[decorate,decoration={brace,mirror,amplitude=5pt}]
        (0,-0.3) -- (9.6,-0.3)
        node[midway,below=7pt] {$q^2\text{ blocks}$};
\end{tikzpicture}
\caption{The concept $c_{b,A}$. Rows and columns within each block
are indexed by $x,y\in\mathbb F_2^q$, respectively.
Only block $b$ is nonzero. For $c_0$, every block is zero.}
\label{fig:bilinear-concept-class}
\end{figure}
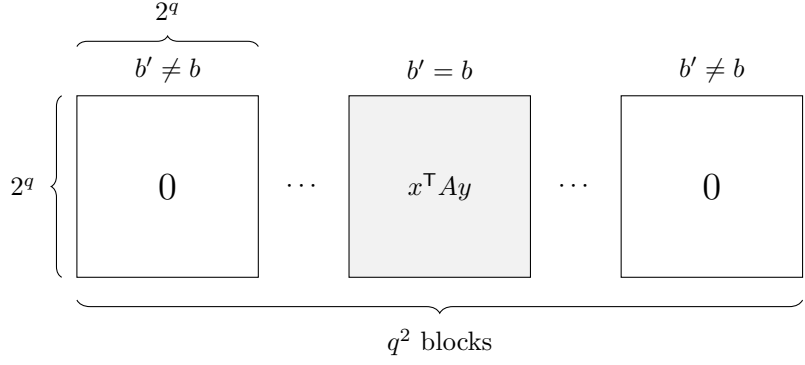

\subsection{Quantum upper bound}\label{sec:dqupper}

We first show that $\Cc_q$ can be learned with $O(q)$ quantum membership queries. The algorithm has two stages: first find the unique nonzero block, and then recover the matrix $A$.

Suppose first that the target concept is $c_{b,A}$ with $A\neq 0$. Since $\ker(A)$ is a proper subspace of $\mathbb F_2^q$, a uniformly random $y\in\mathbb F_2^q$ satisfies $Ay\neq 0$ with probability at least $1/2$. For every such $y$, exactly half of the vectors $x\in\mathbb F_2^q$ satisfy $x^\top Ay=1$. Therefore at least a $1/4$ fraction of the coordinates in block $b$ have value $1$. Since all other $q^2-1$ blocks are identically zero, at least a $1/(4q^2)$ fraction of the entire domain has value $1$. By quantum search~\cite{BHMT02}, we can therefore find a coordinate of value $1$ using $O(q)$ queries with probability at least $2/3$. The block containing this coordinate must be $b$. If the target is $c_0$, quantum search reports failure, and we output $c_0$.

It remains to recover $A$ once $b$ is known. For each $j\in[q]$, fix $y=e_j$. Then
\[
x\mapsto c_{b,A}(b,x,e_j)
= x^\top A e_j
= \langle x,Ae_j\rangle.
\]
Thus this is precisely a Bernstein-Vazirani instance with hidden string $Ae_j$, the $j$th column of $A$. One query suffices to therefore recover $Ae_j$ exactly. Repeating this for all $j\in[q]$ recovers $A$ using $q$ queries. This uses $O(q)+q=O(q)$ queries in total and succeeds with
probability at least $2/3$ on every concept.

\subsection{Deterministic lower bound}\label{sec:dqlower}
Next we give our deterministic lower bound of $q^4$. Fix a correct deterministic learner of strictly smaller cost and consider an adversary that answers $0$ to every query; in particular, these answers are consistent with $c_0$. Suppose the learner halts after fewer than $q^4$ queries. For each $b\in[q^2]$, let $t_b$ be the number of its queries to block $b$. Since $\sum_b t_b<q^4$, some block $b^\star$ satisfies $t_{b^\star}<q^2$.

Write the queries to that block as $(b^\star,x_1,y_1),\ldots,(b^\star,x_{t_{b^\star}},y_{t_{b^\star}})$.
A matrix $A$ consistent with their zero answers must satisfy
\[
x_i^{\mathsf T}Ay_i=0
\]
for every $i\in[t_{b^\star}]$. These are fewer than $q^2$ homogeneous linear constraints on the $q^2$ entries of $A$, so there exists a nonzero matrix $A^\star$ satisfying all of them.

The concept $c_{b^\star,A^\star}$ therefore answers $0$ to every query in block $b^\star$, and answers $0$ to every query outside that block by definition. Hence the learner receives exactly the same transcript on $c_{b^\star,A^\star}$ as on $c_0$, and, being deterministic, produces the same output on both concepts. It cannot identify both correctly, which is a contradiction.

\subsection{Putting it all together}

\dqseparation*

\begin{proof}
Let $\Cc_q$ be as in \Cref{def:bilinear-class}. Observe that the domain size of $\Cc_q$ satisfies  $N_q=q^2\,2^{2q}$, so $\log N_q=\Theta(q)$. Using our quantum upper bound of $\Q(\Cc_q) = O(q)$ from \Cref{sec:dqupper} and $\log N_q=\Theta(q)$, we have $\Q(\Cc_q)^3\log N_q=O(q^4)$. The theorem now follows from our deterministic lower bound of $\D(\C_q)\ge q^4$ from \Cref{sec:dqlower}.
\end{proof}

This matches the general simulation upper bound
$\D(\C)=O(\Q(\C)^3\log N)$ of Servedio and
Gortler~\cite[Theorem~1]{servedio2004equivalences} for this family.

\section{Further results}
\subsection{Booleanization of exact learning}
\label{sec:booleanization}
For every concept class, every Boolean decision about an unknown concept can be solved by first identifying the concept, and then evaluating the desired Boolean function. We show in this section that, quantumly, the converse holds up to a constant factor: every concept class admits a Boolean decision whose query complexity matches that of exact learning.

For any concept class \(\mathcal C\subseteq\{0,1\}^{N}\) and any set \(\mathcal P\subseteq \mathcal C\), define a ``Booleanization'' of $\mathcal C$ as the Boolean function that identifies an element of $\cal P$, i.e., $b_{\mathcal P}(c)=\mathbf 1[c\in \mathcal P]$. Also define
\[
    \Q_{\rm bool}(\mathcal C)=\max_{\mathcal P\subseteq \mathcal C}\Q(b_{\mathcal P})
    \qquad\text{and}\qquad
    \R_{\rm bool}(\mathcal C)=\max_{\mathcal P\subseteq \mathcal C}\R(b_{\mathcal P}).
\]

We use the general adversary bound~\cite{lee2011quantum} given below.  For a function $f:D\to E$ with $D\subseteq\{0,1\}^N$, define
\begin{equation}\label{eq:adv}
\operatorname{Adv}^{\pm}(f)
:=
\max_{\Gamma\neq 0}
\frac{\|\Gamma\|}
{\max_{i\in[N]}\|\Gamma\circ\Delta_i\|},
\end{equation}
where the maximum is over all real symmetric matrices $\Gamma$ indexed by $D\times D$ such that
$\Gamma_{x,y}=0$ whenever $f(x)=f(y)$, and $(\Delta_i)_{x,y}:=\mathbf 1[x_i\neq y_i]$.
Then $\Q(f)=\Theta\!\left(\operatorname{Adv}^{\pm}(f)\right)$.

\begin{theorem}[Quantum Booleanization]
\label{thm:quantum-booleanization}
For every concept class \(\Cc \subseteq \zone^N\),
\[
    \Q(\mathcal C)=\Theta\!\left(\Q_{\rm bool}(\mathcal C)\right).
\]
\end{theorem}

\begin{proof}
    The inequality \(\Q_{\rm bool}(\mathcal C)\leq \Q(\mathcal C)\) is immediate. Let $\mathcal P^*$ be a partition for which $\Q_{\rm bool}(\mathcal C) = \Q(b_{\mathcal P^*})$. An algorithm that identifies \(c\) can also determine whether \(c\in \mathcal P^*\) without making additional queries.    

    For the converse, we use the general adversary bound~\cite{lee2011quantum}.
Let $\operatorname{Id}: \mathcal C \to \mathcal C$ be the function $\operatorname{Id}(c) = c$. Let \(\Gamma\) be an optimal adversary matrix for \(\operatorname{Id}\). By the constraints on $\Gamma$ given after \Cref{eq:adv}, its diagonal will be zero.  Replacing \(\Gamma\) by \(-\Gamma\), if necessary, choose a unit vector \(v\) such that 
\begin{equation}\label{eq:vstar}
    v^*\Gamma v=\lVert\Gamma\rVert.
\end{equation}
    
    For a sign vector \(s\in\{-1,1\}^{\Cc}\), let $S=\operatorname{diag}(s)$ and $\Gamma_s=\frac{\Gamma-S\Gamma S}{2}$. Note that $(\Gamma_s)_{c,c'}=\frac{1-s_cs_{c'}}{2}\Gamma_{c,c'}$, and hence $(\Gamma_s)_{c,c'}=\Gamma_{c,c'}$ if $s_c\neq s_{c'}$, while $(\Gamma_s)_{c,c'}=0$ otherwise. Consequently, if $P_s:=\{c\in\Cc:s_c=1\}$, then $\Gamma_s$ is a valid adversary matrix for $b_{P_s}$.
    We have for $c\neq c'$, $\E_s[s_cs_{c'}]=\E_s[s_c]\E_s[s_{c'}]=0$.
Moreover, $\Gamma_{c,c}=0$ for every $c\in\Cc$. Therefore, for all $c,c'\in\Cc$,
\[
\bigl(\E_s[S\Gamma S]\bigr)_{c,c'}
=
\Gamma_{c,c'}\,\E_s[s_cs_{c'}]
=
0.
\]
Thus $\E_s[S\Gamma S]=0$.
For the choice of $v^*$ from \Cref{eq:vstar}, we have $\mathbb E_s\bigl[v^*\Gamma_s v\bigr]=\frac12 v^*\Gamma v=\frac12\lVert\Gamma\rVert$. By an averaging argument, this implies the existence of \(s\) for which 
\begin{equation}\label{eq:gammaslowerbound}
    \lVert\Gamma_s\rVert \geq v^*\Gamma_s v \geq \frac12\lVert\Gamma\rVert.
\end{equation}

On the other hand, for any $x\in[N]$,
\[
\Gamma_s\circ\Delta_x
=
\frac12\bigl(\Gamma\circ\Delta_x-S(\Gamma\circ\Delta_x)S\bigr),
\]
where we used that $S$ is diagonal. Since conjugation by $S$ preserves operator norm, the triangle inequality gives
\begin{equation}\label{eq:gammasupperbound}
    \|\Gamma_s\circ\Delta_x\|
\le
\frac12\left(
\|\Gamma\circ\Delta_x\|
+
\|S(\Gamma\circ\Delta_x)S\|
\right)
=
\|\Gamma\circ\Delta_x\|.
\end{equation}

For the choice of $s$ from \Cref{eq:gammaslowerbound}, \Cref{eq:gammaslowerbound,eq:gammasupperbound} give
\[
\frac{\|\Gamma_s\|}
{\max_{x\in[N]}\|\Gamma_s\circ\Delta_x\|}
\ge
\frac{\frac12\|\Gamma\|}
{\max_{x\in[N]}\|\Gamma\circ\Delta_x\|}
=
\frac12\,\operatorname{Adv}^{\pm}(\operatorname{Id}),
\]
where the last equality uses our initial assumption that $\Gamma$ is an optimal adversary matrix for $\operatorname{Id}$. Since $\Gamma_s$ is a valid adversary matrix for $b_{\mathcal P_s}$ as justified earlier, it follows that $\operatorname{Adv}^{\pm}(b_{\mathcal P_s}) \ge \frac12\,\operatorname{Adv}^{\pm}(\operatorname{Id})$.
Using the tightness of the general adversary bound~\cite{lee2011quantum},
\[
\Q_{\mathrm{bool}}(\Cc)\ge \Q(b_{\mathcal P_s})
=\Omega\!\left(\operatorname{Adv}^{\pm}(b_{\mathcal P_s})\right)
=\Omega\!\left(\operatorname{Adv}^{\pm}(\operatorname{Id})\right)
=\Omega(\Q(\Cc)),
\]
which proves the claim.
\end{proof}

The analogous statement fails for randomized query complexity, even when every Booleanization is solvable exactly with one query.

\begin{theorem}[Randomized separation]
\label{thm:randomized-booleanization-separation}
For every integer $m\geq2$, there is a concept class $\Cc_m\subseteq\zone^{2^m}$ with $|\Cc_m|=m$ such that
$\R_{\mathrm{bool}}(\Cc_m)=1$ but $\R(\Cc_m)=\Omega(\log m)$.
\end{theorem}

\begin{proof}
Let $\Cc_m=\{c_1,\ldots,c_m\}$ be the class of dictator functions, where each $c_j:\zone^m\to\zone$ is defined by $c_j(z)=z_j$ for every $z\in\zone^m$. Fix an arbitrary set $\mathcal P\subseteq\Cc_m$, and define $z^{\mathcal P}\in\{0,1\}^m$ by $z^{\mathcal P}_j=\mathbf 1[c_j\in\mathcal P]$. If the unknown concept is
$c_j$, then
\[
c_j(z^{\mathcal P})
=
z^{\mathcal P}_j
=
\mathbf 1[c_j\in\mathcal P]
=
b_{\mathcal P}(c_j).
\]
Thus every Booleanization can be computed deterministically with one query, and hence $\R_{\mathrm{bool}}(\Cc_m)=1$.
On the other hand, each membership query returns only one bit. A standard information-theoretic lower bound (see, e.g.,~\cite[Lemma~8]{servedio2004equivalences}) implies $\R(\Cc_m)=\Omega(\log m)$.
\end{proof}

\subsection{Fractional combinatorial parameters}\label{sec:frac}
In this section we first define the standard parameters $\gamma$ and $\etd$ from the literature and then go on to define their fractional analogs.

\subsubsection{Existing parameters}
Inspired by notation from~\cite{servedio2004equivalences}, we use the splitting parameter
\begin{equation*}
\gamma(\mathcal C):=
\min_{\substack{\mathcal S\subseteq\mathcal C\\|\mathcal S|\ge 2}}
\max_{i\in[N]}
\frac{\min\cbra{|\mathcal S_{i,0}|,|\mathcal S_{i,1}|}}{|\mathcal S|}.
\end{equation*}
Some intuition behind this definition is that if a learner maintains the set $\mathcal S$ of concepts consistent with the answers seen so far, then whenever $|\mathcal S|\geq 2$, there is a query such that, regardless of the answer received, at least a $\gamma(\mathcal C)$ fraction of the concepts in $\mathcal S$ are eliminated.

The definition of $\gamma$ can equivalently be viewed distributionally: for each nontrivial subclass $\cal S\subseteq\cal C$, place the uniform distribution on the concepts in $\cal S$ and zero mass on concepts outside $\cal S$.~Then
\begin{equation}\label{eq:gamma}
\gamma(\cal C)=\min_{\substack{\cal S\subseteq\cal C\\ |\cal S|\geq 2}}
\max_{i\in[N]}\min\{\Pr_{c\sim\cal S}[c_i=0],\Pr_{c\sim\cal S}[c_i=1]\}.
\end{equation}
We will also use the density parameter of Bshouty and Makhoul~\cite[Equation~(9)]{BshoutyMakhoul18}, defined~by
\[
\operatorname{DEN}(\mathcal C):=
\max_{\substack{\mathcal S\subseteq\mathcal C\\|\mathcal S|\ge 2}}
\frac{|\mathcal S|-1}
{\displaystyle\max_{i\in[N]}\min\{|\mathcal S_{i,0}|,|\mathcal S_{i,1}|\}}.
\]
We first observe that $\operatorname{DEN}(\cal C)$ is asymptotically equivalent to $1/\gamma(\cal C)$.

\begin{proposition}\label{prop:den-gamma}
For every concept class $\mathcal C$, $\operatorname{DEN}(\mathcal C)
\leq \frac{1}{\gamma(\mathcal C)}
\leq 2\operatorname{DEN}(\mathcal C)$.
\end{proposition}

\begin{proof}
Writing
$m(\mathcal S):=\max_{i\in[N]}\min\{|\mathcal S_{i,0}|,|\mathcal S_{i,1}|\}$,
we have
\[
\frac{1}{\gamma(\mathcal C)}
=\max_{\substack{\mathcal S\subseteq\mathcal C\\|\mathcal S|\ge2}}
\frac{|\mathcal S|}{m(\mathcal S)},
\qquad
\operatorname{DEN}(\mathcal C)
=\max_{\substack{\mathcal S\subseteq\mathcal C\\|\mathcal S|\ge2}}
\frac{|\mathcal S|-1}{m(\mathcal S)}.
\]
The claim follows from $|\mathcal S|-1\leq|\mathcal S|\leq2(|\mathcal S|-1)$ for $|\mathcal S|\geq2$.
\end{proof}

For $h\in\{0,1\}^N$, a set $S\subseteq[N]$ is a \emph{specifying set} for $h$ with respect to $\mathcal C$ if at most one concept in $\mathcal C$ agrees with $h$ on every coordinate in $S$. The \emph{extended teaching dimension} of $\mathcal C$ with respect to $h$ is defined as
\[
\operatorname{ETD}(\mathcal C,h)
:=
\min\left\{
|S|:
S\subseteq[N],\
\bigl|\{c\in\mathcal C:c_i=h_i\text{ for every }i\in S\}\bigr|\leq 1
\right\}.
\]
Equivalently, a specifying set can be represented by weights $w_1,\ldots,w_N\in\{0,1\}$, where $w_i=1$ if coordinate $i$ is included in the set. Thus,
\begin{equation}\label{eq:etd}
\etd(\cal C,h)
=
\min_{\substack{c^*\in\cal C\\ w_1,\ldots,w_N\in\{0,1\}}}
\left\{
\sum_{i=1}^N w_i :
\sum_{i:c_i\neq h_i} w_i\geq 1
\text{ for every }c\in\cal C\setminus\{c^*\}
\right\}.
\end{equation}
The \emph{extended teaching dimension} of $\mathcal C$ is $
\operatorname{ETD}(\mathcal C)
:=
\max_{h\in\{0,1\}^N}\operatorname{ETD}(\mathcal C,h)$. The following proposition is thus immediate from Proposition~\ref{prop:den-gamma} and the observation that $\operatorname{DEN} \leq \etd + 1$~\cite[Lemma~6]{BshoutyMakhoul18}.
\begin{proposition}\label{prop:gamma-etd}
    For all concept classes $\cal C \subseteq \zone^N$, we have $\frac{1}{\gamma(\cal C)} = O(\etd(\cal C))$.
\end{proposition}

\subsubsection{Fractional measures}

We now introduce fractional analogues of $\gamma$ and $\etd$. In the distributional formulation of $\gamma$ from Equation~\eqref{eq:gamma}, the distribution is uniform on a nontrivial subclass $\cal S\subseteq\cal C$, and hence every concept has probability at most $1/2$. We generalize this by allowing an arbitrary distribution on $\cal C$ subject to the same condition that no single concept has probability greater than $1/2$. Similarly, in the formulation of $\etd$ from Equation~\eqref{eq:etd}, a specifying set is represented by weights $w_i\in\{0,1\}$ on the coordinates; we generalize this by allowing arbitrary nonnegative weights. We show that these two fractional relaxations are equivalent up to constant factors.

\begin{defi}[Fractional splitting parameter]\label{def:fgamma}
Define
\[
\mathrm{f}\gamma(\mathcal C)
:=
\min_{\substack{\mu\textnormal{ distribution on }\mathcal C\\
\max_{c\in\mathcal C}\mu(c)\leq 1/2}}
\max_{i\in[N]}
\min\left\{
\mu(\mathcal C_{i,0}),
\mu(\mathcal C_{i,1})
\right\}.
\]
\end{defi}

Thus, $\fgamma(\cal C)$ is the largest value such that for every distribution $\mu$ on $\cal C$ with no atom of mass greater than $1/2$, there is a query such that, regardless of the answer received, at least an $\fgamma(\cal C)$ fraction of the probability mass is eliminated.

\begin{proposition}\label{prop:gamma-fgamma}
For every concept class $\mathcal C \subseteq \zone^N$,
$\frac{1}{\gamma(\mathcal C)}
\leq
\frac{1}{\mathrm{f}\gamma(\mathcal C)}$.
\end{proposition}

\begin{proof}
For every $\mathcal S\subseteq\mathcal C$ with $|\mathcal S|\geq2$, the uniform distribution on $\mathcal S$ is feasible in the definition of $\mathrm{f}\gamma(\mathcal C)$. Hence
$\mathrm{f}\gamma(\mathcal C)\leq\gamma(\mathcal C)$.
\end{proof}

\begin{defi}[Fractional extended teaching dimension]\label{def:fetd}
For $h\in\{0,1\}^N$, define
\[
\operatorname{fETD}(\mathcal C,h) := \min_{\substack{c^*\in\mathcal C\\w_1,\ldots,w_N\geq 0}} \left\{\sum_{i=1}^N w_i: \sum_{i:c_i\neq h_i}w_i\geq 1 \text{ for every }c\in\mathcal C\setminus\{c^*\} \right\}.
\]
Define $\operatorname{fETD}(\mathcal C) := \max_{h\in\{0,1\}^N}\operatorname{fETD}(\mathcal C,h)$.
\end{defi}

\begin{proposition}\label{prop:fetd-etd}
For every concept class $\mathcal C \subseteq \zone^N$,
$\operatorname{fETD}(\mathcal C) \leq \operatorname{ETD}(\mathcal C)$.
\end{proposition}

\begin{proof}
Fix $h\in\{0,1\}^N$, and let $\cal S\subseteq[N]$ be a minimum specifying set for $h$ with respect to $\cal C$. Set $w_i=1$ for $i\in\cal S$ and $w_i=0$ otherwise. If there is a concept in $\cal C$ that agrees with $h$ on every coordinate in $\cal S$, choose it as $c^*$; otherwise choose $c^*$ arbitrarily from $\cal C$. Since $\cal S$ is a specifying set, every $c\in\cal C\setminus\{c^*\}$ differs from $h$ on at least one coordinate in $\cal S$, and hence
$\sum_{i:c_i\neq h_i} w_i\geq 1$.
Thus the above weights form a feasible solution for $\fetd(\cal C,h)$ of value $|\cal S|=\etd(\cal C,h)$. Therefore $\fetd(\cal C,h)\leq \etd(\cal C,h)$, and taking the maximum over $h$ proves the claim.
\end{proof}

We next show that these two fractional relaxations coincide up to a constant factor. More precisely, the reciprocal of $\fgamma(\cal C)$ and $\fetd(\cal C)$ differ by at most a factor of two.

\begin{lemma}\label{lemma:fgamma-fetd}
For every concept class $\mathcal C \subseteq \zone^N$ with $|\mathcal C|\geq2$,
\[
\operatorname{fETD}(\mathcal C)
\leq
\frac{1}{\mathrm{f}\gamma(\mathcal C)}
\leq
2\operatorname{fETD}(\mathcal C).
\]
\end{lemma}

\begin{proof}
Fix $h\in\{0,1\}^N$, and define 
\begin{equation}\label{eq:betah}
    \beta_h:=\min_{\mu}\max_{i\in[N]}\mu(\{c\in\cal C:c_i\neq h_i\}),
\end{equation}
where the minimum is over distributions $\mu$ on $\cal C$ satisfying $\max_c\mu(c)\leq 1/2$. For a fixed $\mu$, choosing each $h_i$ to be the more likely bit at coordinate $i$ makes
$\mu(\{c\in\cal C:c_i\neq h_i\})=\min\{\mu(\cal C_{i,0}),\mu(\cal C_{i,1})\}$.
Hence
\[
\fgamma(\cal C)
=
\min_{\mu}\max_{i\in[N]}\min\{\mu(\cal C_{i,0}),\mu(\cal C_{i,1})\}
=
\min_{\mu}\min_{h\in\{0,1\}^N}\max_{i\in[N]}\mu(\{c\in\cal C:c_i\neq h_i\})
=
\min_{h\in\{0,1\}^N}\beta_h.
\]

We have therefore shown that $\fgamma(\cal C)=\min_h\beta_h$. It remains to compare $1/\beta_h$ with $\fetd(\cal C,h)$ for each fixed $h$, which we do by proving the two required inequalities separately.

\begin{itemize}
    \item We first show $\fetd(\cal C,h)\leq 1/\beta_h$. We have
    \begin{align}
\beta_h
&=
\min_{\mu}
\max_{i\in[N]}
\mu(\{c\in\cal C:c_i\neq h_i\})\nonumber\\
&=
\min_{\mu}
\max_{p}
\sum_{i\in[N]}
p_i\,\mu(\{c\in\cal C:c_i\neq h_i\})\nonumber\\
&=
\max_{p}
\min_{\mu}
\sum_{i\in[N]}
p_i\,\mu(\{c\in\cal C:c_i\neq h_i\}),\label{eq:betah_minimax}
\end{align}
where $p$ ranges over distributions on $[N]$, and $\mu$ ranges over feasible distributions on $\cal C$. The second equality follows since the maximizing $p$ may be taken to be a point mass, and the last equality follows from von Neumann's minimax theorem.

Fix a maximizing distribution $p$, and let $c_1,c_2\in\cal C$ attain the two smallest values of $\sum_{i:c_i\neq h_i}p_i$, in increasing order. Since every feasible distribution $\mu$ satisfies $\mu(c)\leq 1/2$ for every $c\in\cal C$, the minimizing $\mu$ must place probability exactly $1/2$ on each of $c_1$ and $c_2$. By Equation~\eqref{eq:betah_minimax},
\[
\beta_h=\frac{1}{2}\left(\sum_{i:(c_1)_i\neq h_i}p_i+\sum_{i:(c_2)_i\neq h_i}p_i\right)\leq \sum_{i:(c_2)_i\neq h_i}p_i.
\]
Since $c_2$ attains the second-smallest value, it follows that $\sum_{i:c_i\neq h_i}p_i\geq\beta_h$ for every $c\in\cal C\setminus\{c_1\}$. Choosing $c^*=c_1$ and setting $w_i:=p_i/\beta_h$ in Definition~\ref{def:fetd} therefore gives a feasible solution for $\fetd(\cal C,h)$ of total weight $1/\beta_h$. Thus $\fetd(\cal C,h)\leq 1/\beta_h$.
\item We next show that $1/\beta_h\leq 2\fetd(\cal C,h)$. Let $c^*$ and $w_1,\ldots,w_N$ be an optimal solution to $\fetd(\cal C,h)$ using the notation from Definition~\ref{def:fetd}. In particular, we have
\begin{equation}\label{eq:sumwi}
\sum_i w_i=\fetd(\cal C,h).
\end{equation}
Fix any distribution $\mu$ that is feasible in the definition of $\beta_h$ in Equation~\eqref{eq:betah}. Since $\mu(c^*)\leq 1/2$, we have $\mu(\cal C\setminus\{c^*\})\geq 1/2$. Thus,
\begin{align*}
\sum_{i\in[N]} w_i\,\mu(\{c\in\cal C:c_i\neq h_i\})
&=
\sum_{c\in\cal C}\mu(c)\sum_{i:c_i\neq h_i}w_i\geq
\sum_{c\in\cal C\setminus\{c^*\}}\mu(c)
\geq \frac12,
\end{align*}
where the first inequality follows from the constraints in Definition~\ref{def:fetd}, which give $\sum_{i:c_i\neq h_i}w_i\geq 1$ for every $c\neq c^*$, and the last inequality follows from $\mu(c^*)\leq 1/2$.
By Equation~\eqref{eq:sumwi}, the LHS above is a weighted average of the quantities
$\mu(\{c\in\cal C:c_i\neq h_i\})$ with total weight $\fetd(\cal C,h)$. Hence there exists some $i\in[N]$ such that
$\mu(\{c\in\cal C:c_i\neq h_i\})\geq 1/(2\fetd(\cal C,h))$.
Since this holds for every distribution $\mu$ feasible in Equation~\eqref{eq:betah}, we obtain
$\beta_h\geq 1/(2\fetd(\cal C,h))$, or equivalently $1/\beta_h\leq 2\fetd(\cal C,h)$.
\end{itemize}

Thus $\fetd(\cal C,h)\leq1/\beta_h\leq2\fetd(\cal C,h)$ for every $h$. Taking the maximum over $h$, and using $1/\fgamma(\cal C)=\max_h1/\beta_h$, proves the claim.
\end{proof}

The results in Propositions~\ref{prop:gamma-fgamma},~\ref{prop:fetd-etd}, and Lemma~\ref{lemma:fgamma-fetd} are summarized in Theorem~\ref{thm:frac-equivalence} below.
\begin{theorem}\label{thm:frac-equivalence}
    For all concept classes $\cal C \subseteq \zone^N$,
    \[
    \frac{1}{\gamma(\mathcal C)}
\leq
\frac{1}{\mathrm{f}\gamma(\mathcal C)},
\qquad
\operatorname{fETD}(\mathcal C)
\leq
\operatorname{ETD}(\mathcal C),
\qquad
\operatorname{fETD}(\mathcal C)
\leq
\frac{1}{\mathrm{f}\gamma(\mathcal C)}
\leq
2\operatorname{fETD}(\mathcal C).
    \]
\end{theorem}

\subsubsection{Bounds in terms of fractional measures}
We now relate these fractional parameters to the query complexity of exact learning. We first show that $\operatorname{fETD}$ lower bounds quantum query complexity quadratically.

\begin{theorem}\label{thm:quantum-fgamma}
For every concept class $\cal C$,
\[
\Q(\cal C)=\Omega\!\left(\sqrt{\frac{1}{\fgamma(\cal C)}}\right).
\]
\end{theorem}

\begin{proof}
Let $\mu$ be a distribution attaining the minimum in the definition of $\fgamma(\cal C)$ (Definition~\ref{def:fgamma}). Thus $\mu(c)\leq 1/2$ for every $c\in\cal C$, and 
\begin{equation}\label{eq:fgamma-distribution}
\min\{\mu(\cal C_{i,0}),\mu(\cal C_{i,1})\}\leq \fgamma(\cal C) \qquad \forall i \in [N].    
\end{equation}
Consider the matrix $\Gamma$ whose rows and columns are indexed by $\cal C\times\cal C$, defined by
\[
\Gamma[c,c']
=
\begin{cases}
\sqrt{\mu(c)\mu(c')} & c\neq c',\\
0 & c=c'.
\end{cases}
\]
Let $v\in\mathbb{R}^{\cal C}$ be given by $v_c=\sqrt{\mu(c)}$. Since $\|v\|=1$,
\begin{align*}
v^\top\Gamma v
&=
\sum_{\substack{c,c'\in\cal C\\ c\neq c'}}
v_c\,\Gamma_{c,c'}\,v_{c'}
=
\sum_{\substack{c,c'\in\cal C\\ c\neq c'}}
\mu(c)\mu(c')
=
\left(\sum_{c\in\cal C}\mu(c)\right)^2 - \sum_{c\in\cal C}\mu(c)^2 = 1-\sum_{c\in\cal C}\mu(c)^2\\
&\geq 1 - \rbra{\max_{c \in \cal C}\mu(c)}\sum_{c \in \cal C}\mu(c) \geq 1/2.
\end{align*}
Since $\Gamma$ is a real symmetric matrix, we have $\|\Gamma\| = \max_{\norm{x} = 1}|x^\top \Gamma x| \geq |v^\top \Gamma v| \geq 1/2$.

Fix $i\in[N]$, and let $\Delta_i$ denote the matrix with $(c,c')$-entry equal to $1$ iff $c_i\neq c'_i$. Order the concepts so that those in $\cal C_{i,0}$ come first and those in $\cal C_{i,1}$ come second. Since $\Delta_i(c,c')=1$ exactly when $c_i\neq c'_i$, we can write
\[
\Gamma\circ\Delta_i
=
\begin{pmatrix}
0 & ab^\top\\
ba^\top & 0
\end{pmatrix},
\]
where $a_c=\sqrt{\mu(c)}$ for $c\in\cal C_{i,0}$ and $b_c=\sqrt{\mu(c)}$ for $c\in\cal C_{i,1}$.
Since $\Gamma\circ\Delta_i$ is symmetric, its operator norm is $\|ab^\top\|=\|a\|\,\|b\|=\sqrt{\mu(\cal C_{i,0})\mu(\cal C_{i,1})} \leq \sqrt{\fgamma(\cal C)}$ by Equation~\eqref{eq:fgamma-distribution}. Therefore, by the positive-weight adversary bound~\cite{Ambainis02},
\[
\Q(\cal C)
=
\Omega\!\left(
\frac{\|\Gamma\|}{\max_{i\in[N]}\|\Gamma\circ\Delta_i\|}
\right)
=
\Omega\!\left(\frac{1}{\sqrt{\fgamma(\cal C)}}\right).
\]
\end{proof}

We next show that $\operatorname{fETD}$ also gives a randomized upper bound. The proof is inspired by the entropy-based argument of Arunachalam et al.~\cite{arunachalam2021two}, but replaces their use of the nonnegative adversary bound by the fractional extended teaching dimension.
\begin{theorem}\label{thm:r-fetd}
For every concept class $\mathcal C \subseteq \zone^N$,
\[
\mathsf R(\mathcal C)
=
O\!\left(
\frac{\operatorname{fETD}(\mathcal C)\log |\cal{C}|}
{\log\bigl(1+\operatorname{fETD}(\mathcal C)\bigr)}
\right).
\]
\end{theorem}
\begin{proof}
We use Yao's minimax principle. Fix an arbitrary distribution $\mu$ on $\cal C$. We will construct a deterministic decision tree that has error at most $1/6$ under $\mu$ and expected depth at most
\[
d=O\!\left(
\frac{\fetd(\cal C)}{1 + \log(\fetd(\cal C))}
\log|\cal C|
\right).
\]
Truncating this tree after $6d$ queries and outputting an arbitrary concept if the tree has not already stopped gives a deterministic tree of depth $6d$. By Markov's inequality,
\[
\Pr[T>6d]\leq \frac{\mathbb E[T]}{6d}\leq \frac16,
\]
where $T$ is a random variable denoting the depth of the original tree on an input drawn from $\mu$. Hence the truncated deterministic decision tree has error at most $1/6+1/6=1/3$ under the input distribution $\mu$. By Yao's minimax principle, it therefore suffices to prove the stated bound on the expected depth of the tree constructed below.

The construction of the tree is straightforward. At a node $v$, let $\nu_v$ denote the distribution on $\cal C$ obtained by conditioning $\mu$ on the transcript leading to $v$. If there exists a $c$ with $\nu_v(c) \geq 5/6$, the tree stops and outputs $c$. At every leaf $v$, the tree outputs a concept $c$ satisfying $\nu_v(c)\geq 5/6$. Hence, conditioned on reaching $v$, the error probability is at most $1/6$. Averaging over the leaves, the overall error of the tree under $\mu$ is at most $1/6$ as required. At an internal node, query a coordinate
\begin{equation}\label{eq:iv}
    i_v = \argmax_{i\in[N]}\min\{\nu_v(\cal C_{i,0}),\nu_v(\cal C_{i,1})\}.
\end{equation}
That is, $i_v$ is a coordinate whose query guarantees the largest possible decrease in the posterior mass of the remaining concepts, regardless of the answer received. The remainder of the proof analyzes the expected depth of this tree.

Fix an arbitrary non-leaf node $v$, and define $h\in\{0,1\}^N$ by choosing, for each $i\in[N]$, a bit $h_i$ satisfying $\nu_v(\cal C_{i,h_i})\geq \nu_v(\cal C_{i,1-h_i})$. In particular, this means
\begin{equation}\label{eq:minorityh}
\nu_v(\{c\in\cal C:c_i\neq h_i\})=\min\{\nu_v(\cal C_{i,0}),\nu_v(\cal C_{i,1})\}~\forall i\in[N].
\end{equation}
Let $c^*\in\cal C$ and $w_1,\ldots,w_N\geq 0$ attain $\fetd(\cal C,h)$ (see Definition~\ref{def:fetd}). Since $v$ is a non-leaf node, $\nu_v(c^*)<5/6$. Taking the $w_i$-weighted sum of the minority masses, we obtain
\begin{align}
\sum_{i\in[N]}w_i\min\{\nu_v(\cal C_{i,0}),\nu_v(\cal C_{i,1})\}
& =
\sum_{i\in[N]}w_i\,\nu_v(\{c\in\cal C:c_i\neq h_i\})
=
\sum_{c\in\cal C}\nu_v(c)\sum_{i:c_i\neq h_i}w_i\nonumber\\
&=
\sum_{c\neq c^*}\nu_v(c)\sum_{i:c_i\neq h_i}w_i
+
\nu_v(c^*)\sum_{i:c_i^*\neq h_i}w_i\nonumber\\
&\geq
\sum_{c\neq c^*}\nu_v(c)
=
1-\nu_v(c^*) >
\frac16, \label{eq:weighted-nuv}
\end{align}
where the first equality follows from Equation~\eqref{eq:minorityh}, the second equality is obtained by changing the order of summation, the first inequality follows from the defining constraints of $\fetd(\cal C,h)$ together with the nonnegativity of the weights (and hence nonnegativity of the second term on the second line), and the final equality uses that $\nu_v$ is a distribution on $\cal C$. The last inequality follows since $\nu_v(c^*)<5/6$ which holds since $v$ is a non-leaf node.

Since $\sum_i w_i=\fetd(\cal C,h)$, Equation~\eqref{eq:weighted-nuv} implies existence of $i\in[N]$ such that
\[
\min\{\nu_v(\cal C_{i,0}),\nu_v(\cal C_{i,1})\}
>
\frac{1}{6\fetd(\cal C,h)}
\geq
\frac{1}{6\fetd(\cal C)}.
\]
By the definition of $i_v$ in Equation~\eqref{eq:iv}, we therefore have
\begin{equation}\label{eq:ivbound}
\min\{\nu_v(\cal C_{i_v,0}),\nu_v(\cal C_{i_v,1})\}>\frac{1}{6\fetd(\cal C)}.
\end{equation}
We now analyze the entropy decrease caused by the query at a non-leaf node $v$. Let $C$ denote the random concept distributed according to $\nu_v$, and let $B:=C_{i_v}$ be the answer to the query at $v$. Then the expected posterior entropy after the query is $H(C\mid B)$, and hence the expected entropy decrease is
\[
H(C)-H(C\mid B)=H(B),
\]
where the equality uses that $B$ is a deterministic function of $C$ since the tree is deterministic. 
Let $q_v:=\min\{\nu_v(\cal C_{i_v,0}),\nu_v(\cal C_{i_v,1})\} \leq 1/2$.
Then $B$ is a Bernoulli random variable with probabilities $q_v$ and $1-q_v$, and hence $H(B)=H_2(q_v)$. By Equation~\eqref{eq:ivbound}, we have $1/(6\fetd(\cal C))<q_v\leq 1/2$. Since binary entropy is increasing on $[0,1/2]$ and since $H_2(p) \geq p \log(1/p)$ for all $p$,
\begin{equation}\label{eq:hb-large}
H(B)=H_2(q_v)\geq H_2\!\left(\frac{1}{6\fetd(\cal C)}\right)
=
\Omega\!\left(\frac{\log(\fetd(\cal C)+1)}{\fetd(\cal C)}\right).
\end{equation}

Let $T$ be the random variable denoting the depth at which the tree stops on a random concept $C\sim\mu$. For $t\geq 0$, let $Z_t$ denote the transcript after the first $\min\{t,T\}$ queries, and define 
\begin{equation}\label{eq:ht-def}
H_t:=H(C\mid Z_t)
=
\sum_z \Pr[Z_t=z]\,H(C\mid Z_t=z).
\end{equation}
Let
\begin{equation}\label{eq:delta-def}
\delta:=\Omega\!\left(\frac{\log(\fetd(\cal C)+1)}{\fetd(\cal C)}\right)
\end{equation}
be the lower bound from Equation~\eqref{eq:hb-large} on the expected entropy decrease at every non-leaf node.

Fix $t\geq 0$ and condition on a particular value $Z_t=z$. If $z$ is a transcript that leads to a leaf, then no further query is made and hence $Z_{t+1}=Z_t$, so the entropy decrease is $0$. If $z$ is a transcript that leads to an internal node, then the tree makes a query at this node, and by Equation~\eqref{eq:hb-large} the expected entropy decrease from this query, conditioned on $Z_t=z$, is at least $\delta$.
Thus, for every possible transcript $z$,
\[
H(C\mid Z_t=z)
-
\mathbb E[H(C\mid Z_{t+1})\mid Z_t=z]
\geq
\begin{cases}
0 & \text{if $z$ leads to a leaf},\\
\delta & \text{if $z$ leads to an internal node}.
\end{cases}
\]
Averaging over all possible values of $Z_t$ and recalling the definition of $H_t$ from Equation~\eqref{eq:ht-def}, we therefore obtain
\begin{align*}
H_t-H_{t+1}
&=
\sum_z \Pr[Z_t=z]
\left(
H(C\mid Z_t=z)
-
\mathbb E[H(C\mid Z_{t+1})\mid Z_t=z]
\right)\\
&\geq
\delta
\sum_{\substack{z:\,z\text{ leads to}\\ \text{an internal node}}}
\Pr[Z_t=z] =
\delta\,\Pr[T>t].
\end{align*}
Summing over $t\geq 0$, we obtain
\[
\log|\cal C|
\geq H(C)=H_0
\geq H_0-\lim_{t\to\infty}H_t
=
\sum_{t\geq0}(H_t-H_{t+1})
\geq
\delta\sum_{t\geq0}\Pr[T>t]
=
\delta\,\mathbb E[T].
\]
The last equality uses the standard identity
$\mathbb E[T]=\sum_{t\geq0}\Pr[T>t]$
for every nonnegative integer-valued random variable $T$.
Plugging back in the value of $\delta$ from Equation~\eqref{eq:delta-def}, we have
\[
\mathbb E[T]
=
O\!\left(
\frac{\fetd(\cal C)}{\log(\fetd(\cal C)+1)}
\log|\cal C|
\right),
\]
which is precisely the expected-depth bound required at the beginning of the proof.
\end{proof}

Lemma~\ref{lemma:fgamma-fetd}, Theorems~\ref{thm:quantum-fgamma} and~\ref{thm:r-fetd} immediately yield Theorem~\ref{thm:qr-fetd}.

\bibliography{bibo}

\end{document}